\documentclass[11pt]{amsart}
\usepackage{graphicx}
\usepackage[margin=1.0in]{geometry}
\usepackage{setspace}
\usepackage{amsthm,amsmath,epsfig,subfigure,amsfonts,amssymb,bm,xcolor,amsbsy}
\usepackage[colorlinks,linkcolor=blue,citecolor=blue,anchorcolor=blue]{hyperref}

\numberwithin{equation}{section}
\theoremstyle{plain}

\newcommand{\beq}[1]{\begin{equation} \label{#1}}
\newcommand{\eeq}{\end{equation}}
\newcommand{\bed}{\begin{displaymath}}
\newcommand{\eed}{\end{displaymath}}
\newcommand{\ben}{\begin{eqnarray*}}
\newcommand{\een}{\end{eqnarray*}}
\allowdisplaybreaks[1]

\def\cd{(\cdot)}

\newcommand{\wdt}{\widetilde}

\newtheorem{thm}{Theorem}[section]
\newtheorem{prop}[thm]{Proposition}
\newtheorem{lem}[thm]{Lemma}
\newtheorem{cor}[thm]{Corollary}

\theoremstyle{definition}
\newtheorem{rem}[thm]{Remark}

\newcommand{\bedd}{\bed\begin{array}{l}}
\newcommand{\eedd}{\end{array}\eed}

\def\one{{\hbox{1{\kern -0.35em}1}}}

\newcommand{\bea}{\bed\begin{array}{rl}}
\newcommand{\eea}{\end{array}\eed}

\newcommand{\barray}{\begin{array}{ll}}
\newcommand{\earray}{\end{array}}

\begin{document}

\title[]{Optimal Dividend Strategy for an Insurance Group with Contagious Default Risk}

\author[  ]{Zhuo Jin}
\address[Z. Jin]{Centre for Actuarial Studies, Department of Economics, The University of Melbourne, VIC3010, Australia.}
\email{zhuo.jin@unimelb.edu.au}

\author[  ]{Huafu Liao}
\address[H. Liao]{School of Mathematical Sciences, University of Science and Technology of China, Hefei, Anhui, 230026, China.}
\email{lhflhf@mail.ustc.edu.cn}

\author[  ]{Yue Yang}
\address[Y. Yang]{Department of Applied Mathematics, The Hong Kong Polytechnic University, Hung Hom, Kowloon, Hong Kong.}
\email{yue.yy.yang@connect.polyu.hk}

\author[  ]{Xiang Yu}
\address[X. Yu]{Department of Applied Mathematics, The Hong Kong Polytechnic University, Hung Hom, Kowloon, Hong Kong.}
\email{xiang.yu@polyu.edu.h}

\maketitle

\begin{abstract}
This paper studies the optimal dividend for a multi-line insurance group, in which each subsidiary runs a product line and is exposed to some external credit risk. The default contagion is considered such that one default event may increase the default probabilities of all surviving subsidiaries. The total dividend problem for the insurance group is investigated and we find that the optimal dividend strategy is still of the barrier type. Furthermore, we show that the optimal barrier of each subsidiary is modulated by the default state. That is, how many and which subsidiaries have defaulted will determine the dividend threshold of each surviving subsidiary. These conclusions are based on the analysis of the associated recursive system of \textit{Hamilton-Jacobi-Bellman} variational inequalities (HJBVIs). The existence of the classical solution is established and the verification theorem is proved. In the case of two subsidiaries, the value function and optimal barriers are given in analytical forms, allowing us to conclude that the optimal barrier of one subsidiary decreases if the other subsidiary defaults.
\vspace{0.4 cm}

\end{abstract}

\ \\

\section{Introduction}\label{sec:intro}

Dividend payment is always a focused issue in insurance and corporate finance, which is regarded as an important signal of the company's future growth opportunities and has direct impact on the wealth of shareholders. Meanwhile, insurance companies also dynamically invest money in the financial market in order to pay future claims. The pioneer work \cite{Finetti} solves the optimal dividend problem up to the financial ruin time when the surplus process follows a simple random walk. Later, vast research has been devoted to finding optimal dividend strategies in various discrete and continuous time risk models, see a short list of related work by  \cite{Gerber72, AsmussenT, HojgaardT99, AzcueM05, AvramPP07, Schmidli, LoeffenR10, ChevalierVS13, Kaz3, Kaz2, Kaz, Perez} and references therein. We refer to \cite{AlbrecherT09} and \cite{Avanzi09} for some comprehensive surveys on the topic of dividend optimization.

The present paper has a particular interest in a multi-line insurance group, which is a parent insurer consisting of multiple subsidiaries in the market where each subsidiary runs a product line such as life insurance, auto insurance, income protection insurance, housing insurance and etc. Each product line is subject to bankruptcy separately and has its own premiums and losses with very distinctive claim frequency, which motivates some recent academic studies on multi-line insurance business. In a multi-line insurance group framework, the insurance pricing model by line is studied by \cite{PhillipsCA98}. The capital allocation strategy for a multi-line insurance company is investigated by \cite{MyersR01}, which reveals that allocations depend on the uncertainty of each line's losses and the marginal contribution of each line. Under the assumption that losses from all product lines follow a sharing rule, some premiums problems are examined by \cite{IbragimovJW10}.

What is missing in the literature is the investigation of external systemic risk for the insurance group. Our work enriches the study of the insurance group by considering the group dividend optimization problem in which each subsidiary may go default due to some contagious default risk. In practice, many subsidiaries share the same reserves pool from the parent group company. It is reasonable to assume that all subsidiaries are exposed to some common credit risk. Our model can depict some real life situations that the group manager collects cash reserves from different subsidiaries and invests them into some financial credit instruments such as defaultable Bonds, CDS, equity default swaps and etc. The insolvency and termination of one subsidiary business caused by the market credit risk may quickly spread to all other subsidiaries if they share the same underlying credit assets. Some empirical studies find that defaults are indeed contagious in certain cases and exhibit the so-called default-clustering phenomenon, see \cite{DasDKS}. In particular, a dependent credit risk model is studied by \cite{TakadaS}, which analyzes the contagious defaults affected by a common macroeconomic factor. A financial network model is later developed by \cite{AminiM}, in which the contagious defaults are caused by a macroeconomic shock. In the context of insurance, it is also reasonable to consider the investment of net-reserves in some credit assets and the default risk in the financial market may lead to some massive domino effects in surplus management and subsidiaries operations.

It is worth noting that some recent work such as \cite{AlbrecherAM}, \cite{GuSZ} and \cite{Grandits} consider the collaborating dividend problem between multiple insurance companies, in which the credit default and default contagion are again not concerned. Instead, they consider some independent insurance companies and assume that one insurance company can inject capital into other companies whenever their financial ruins occur. The optimal dividend for two collaborating insurance companies in compound Poisson and diffusion models are studied by \cite{AlbrecherAM} and \cite{GuSZ} respectively. The extension to different solvency criteria is considered later by \cite{Grandits}. Although these work differ substantially from the present paper,  we confront similar challenges from the multi-dimensional singular control problem and some new mathematical methods are required.

To ensure the tractability, we work in the interacting intensity framework to model default contagion, which allows sequential defaults and assumes that the credit default of one subsidiary can affect other surviving names by increasing their default intensities. This type of default contagion has been actively studied recently in the context of portfolio management, see among \cite{B3, B2, B1, BLY2, BLY1} and many others. The key observation in these work is that the system of HJB partial differential equations (PDEs) is recursive and the depth of the recursion equals the number of risky assets. The system of PDEs can therefore be analyzed using a backward recursion from the state in which all assets are defaulted towards the state that all assets are alive. As opposed to portfolio optimization, we confront a singular control problem that stems from the dividend payment, and we consequently need to handle variational inequalities instead of PDE problems. To the best of our knowledge, our work appears as the first one attempting to introduce the default contagion to the insurance group dividend control framework. In particular, we distinguish the ruin caused by insurance claims (i.e. the surplus process diffuses to zero) and the termination caused by credit default jump. It is observed in this paper that the optimal group dividend is of the barrier type and the optimal barrier for each subsidiary is default-state-modulated, i.e., the optimal barrier of each surviving subsidiary will be adjusted whenever some subsidiaries go default. In the simple case of two subsidiaries, we can rigorously prove that the group manager lowers the dividend barrier of the surviving subsidiary and forces it to pay dividend soon, see Corollary \ref{change-m}.

Our mathematical contribution is the study of the recursive system of HJBVIs \eqref{original HJB1}, which differs from some conventional PDE problems in portfolio optimization. We adopt the core idea in \cite{B3, B2, B1, BLY1} and follow the backward recursion based on the number of defaulted subsidiaries. In addition, we take the full advantage of the risk neutral valuation of the group control and simplify the multi-dimensional value function into a separation form. Our arguments can be outlined as follows. Firstly, we start from the case when there is only one surviving subsidiary and work inductively to the case when all subsidiaries are alive. The classical solution in the step with $k$ surviving subsidiaries will appear as variable coefficients in the step with $k+1$ surviving subsidiaries, and we can continue to show the existence of classical solution with $k+1$ names. Secondly, to show the existence of classical solution in each step with a fixed number of subsidiaries, we conjecture a separation form of the value function, and split the variational inequality from the group control into a subsystem of auxiliary variational inequalities. To tackle each auxiliary variational inequality, we first obtain the existence of a classical solution to the ODE problem. By applying the smooth-fit principle, we deduce the existence of a free boundary point depending on the default state and construct the desired classical solution to the auxiliary variational inequality. The rigorous proof of the verification theorem is provided to show that the value function coincides with the classical solution to the recursive system of HJBVIs \eqref{original HJB1}. As a byproduct, the optimal dividend is proved to be a reflection strategy with the barrier depending on the default state indicator process, see \eqref{optim divi} in Theorem \ref{mainthm1}.

The rest of the paper is organized as follows. Section \ref{sec:form} introduces the model of the multi-line insurance group with external credit default contagion. The optimal group dividend problem for all subsidiaries is formulated and the main theorem is presented therein. In Section \ref{2companies}, we derive the HJBVI \eqref{original HJB} for two subsidiaries and solve the value function in an explicit manner. The optimal barriers of the dividend are constructed using the smooth-fit principle. Section \ref{Multi comp} generalizes the results to a multi-line insurance group. The proof of the verification theorem is given in Section \ref{sec:verf}.
The derivation of the HJBVI \eqref{original HJB} for two subsidiaries is reported in Appendix \ref{appx}.

\section{Model Formulation}\label{sec:form}
Let $(\Omega, \mathcal{F}, \mathbb{F},\mathbb{P})$ be a complete filtered probability space where $\mathbb{F}:=\{\mathcal{F}_t\}$ is a right-continuous, $\mathbb{P}$-completed filtration. We consider an insurance group consisting of $N$ subsidiary business units and each business unit is managed independently within the group. In particular, the decision maker in the present paper is the insurance group manager, who collects the premiums and contributes shares of the dividend for the whole group of subsidiaries.

After the pioneer work \cite{Ig}, the diffusion-approximation of the classical Cram\'er-Lundberg model has been popular in the study of optimal dividend and reinsurance thanks to its tractability and allowance of explicit control strategies, see among \cite{Em75}, \cite{Grand}, \cite{Asmu}, \cite{Choulli}, \cite{Gerber} and many others.  Following their setting, it is assumed in this paper that all subsidiaries have the same form of surplus processes with different drifts and insurance claim distributions and the pre-default surplus process $\hat{X}_i(t)$ for each subsidiary satisfies the diffusion model that
\begin{align}\label{eq:38}
d\hat{X}_i(t)=a_i dt-b_i dW_i(t),
\end{align}
where constants $a_i>0$ and $b_i>0$ represent the mean and the volatility of the surplus process respectively, and each $W_i(t)$ is a standard $\mathbb{P}$-Brownian motion. For $1\leq i,j\leq N$, the correlation coefficient between $W_i$ and $W_j$ is denoted by the constant $-1\leq \rho_{ij}\leq 1$ and the correlation coefficient matrix is denoted by $\Sigma=(\rho_{ij})_{N\times N}$. The model covers correlated insurance claims from different subsidiaries including possible scenarios that some subsidiaries are running product lines that depend on other product lines and some subsidiaries serve certain overlapping customers.

We consider in this paper that each subsidiary allocates a large proportion of its net-reserves in some credit assets. Each subsidiary is exposed to some external credit risk in the financial market, and a wave of defaults in these credit assets may lead to large loss of net-reserves in all subsidiaries. One example is the collapse of AIG, which is exposed to substantial credit risk in its balance sheet in the 2008 financial crisis. To make our multi-dimensional dividend control problem tractable and facilitate the backward induction method, we consider the extreme case in the present paper that the external default will terminate the operation of the subsidiary and no salvage value can be paid as dividend at the moment of default. To model these extreme and irreparable default events, we choose the so-called default indicator process that is described by an N-dimensional $\mathbb{F}$-adapted process $\mathbf{Z}(t)=(Z_1(t),\ldots, Z_N(t))$ taking values on $\{0,1\}^N$. For each $i$, $Z_i(t)=1$ indicates that the $i$-th subsidiary has defaulted up to time $t$, while $Z_i(t)=0$ indicates that the $i$-th subsidiary is still alive at time $t$. The process $\mathbf{Z}(t)$ is assumed to be independent of all Brownian motions $W_i(t)$, $i=1,\ldots,N$, to reflect that these external default events stem from the credit assets and they do not depend on the claims of each subsidiary's insurance products.

For each $i=1,\ldots, N$, the default time $\sigma_i$ for the $i$-th subsidiary is given by
\begin{align*}
\sigma_i:=\inf\left\{t\geq0;Z_i(t)=1\right\}.
\end{align*}
The stochastic intensity of $\sigma_i$ is modeled by $\left(1-Z_i(\cdot)\right)\lambda_i\left(\mathbf{Z}(\cdot)\right)$, where $\lambda_i$ maps $\{0,1\}^N$ to $(0,+\infty)$ and the process
\begin{align}\label{Z_imart}
M_i(t):=Z_i(t)-\int_0^{t\wedge\sigma_i}\lambda_i\left(\mathbf{Z}(s)\right)ds,
\end{align}
is a martingale with respect to the filtration generated by $\mathbf{Z}$. Note that this process $Z_i(t)$ can also be viewed as a Cox process truncated above by constant $1$, whose intensity process is $(1-Z_i(t))\lambda_i(\mathbf{Z}(t))+Z_i(t)$.

Let us take $N=2$ as an example and consider the default state $\mathbf{Z}(t)=(0,0)$ at time $t$. The values $\lambda_1(0,0)$ and $\lambda_2(0,0)$ give the default intensity of subsidiary $1$ and subsidiary $2$ at time $t$ respectively. Suppose that subsidiary $1$ has already defaulted before time $t$ and only subsidiary $2$ is alive, then $\lambda_2(1,0)$ represents the default intensity of subsidiary $2$ at time $t$. Similarly, if the subsidiary $2$ has already defaulted before time $t$ and only subsidiary $1$ is alive, then $\lambda_1(0,1)$ represents the default intensity of subsidiary $1$ at time $t$. Moreover, we consider the default contagion in the sense that $\lambda_1(0,0)\leq \lambda_1(0,1)$ and
$\lambda_2(0,0)\leq \lambda_2(1,0)$ such that the default intensity of one subsidiary increases after the other subsidiary defaults.

For the general case with $N$ subsidiaries, the default indicator process at time $t$ may jump from a state $\mathbf{Z}(t)=(Z_1(t),\ldots,Z_{i-1}(t),Z_{i}(t),Z_{i+1}(t),\ldots, Z_N)$ in which the subsidiary $i$ is alive ($Z_{i}(t)=0$) to the neighbour state $(Z_1(t),\ldots,Z_{i-1}(t),1-Z_{i}(t),Z_{i+1}(t),\ldots, Z_N)$ in which the subsidiary $i$ has defaulted with the stochastic rate $\lambda_i(\mathbf{Z}(t))$. It is assumed from this point on that $Z_i$, $i=1,\ldots, N$, will not jump simultaneously in the sense that
\begin{align}\label{simul_jump_assumption1}
\Delta Z_i(t)\Delta Z_j(t)=0,\quad 1\leq i<j\leq N,\quad t\geq0.
\end{align}
Note that the default intensity of the $i$-th subsidiary $\lambda_i(\mathbf{Z}(t))$ depends on the whole vector process $\mathbf{Z}(t)$, and it is assumed that $\lambda_i(\mathbf{Z}(t))$ increases if any other subsidiary defaults. This is what we mean by default contagion for multiple subsidiaries. Let us denote the vector $\lambda(\mathbf{z}) = (\lambda_i(\mathbf{z}); i=1,\ldots, N )^T$, for the given default vector $\mathbf{z}\in\{0,1\}^N$.

The actual surplus process of subsidiary $i$ after the incorporation of external credit risk is denoted by $\wdt X_i(t)$, where $i=1, 2, \ldots, N,$ and it is defined as
\begin{align}\label{eq:39}
\wdt X_i(t):=\left(1-Z_i(t)\right)\hat{X}_i(t).
\end{align}

Given the surplus process $\wdt X_i(t)$, for each subsidiary $i$, we can then introduce the dividend policy. A dividend strategy $D_i\cd$ is an $\mathcal{F}_t$-adapted
  process representing the accumulated
  amount of dividend paid up to time $t$. That is,
  $D_i(t)$ is a nonnegative and nondecreasing
  stochastic process that is right continuous and have
  left limits with $D_i(0^-)=0$. The jump size of
$D_i$ at time $t\ge 0$ is denoted by $\Delta D_i(t):= D_i(t)-D_i(t^-) $, and
$D_i^c(t) := D_i(t)-\sum_{0\le s \le t} \Delta D_i(s)$ denotes the continuous part of
$D_i(t)$.

For the $i$-th subsidiary, the resulting surplus process in the presence of dividend payments can be written as
\beq{sur_2}
X_i(t):=\big(1-Z_i(t)\big)\big(\wdt X_i(t)-D_i(t)\big), \ X_i(0)=x_i\geq 0,
\eeq
where $x_i$ stands for the initial surplus of the $i$-th subsidiary. We denote the vector process $\mathbf{X}(t):=(X_1(t), \ldots, X_N(t))$.

The objective function for the insurance group is formulated as a corporative singular control of total dividend strategy $\mathbf{D}(t)=(D_1(t),\ldots,D_N(t))$ under the expected value of discounted future dividend payments up to the ruin time
\beq{per_2}
J(\mathbf{x},\mathbf{z},\mathbf{D}\cd):={\mathbb E}\left(\sum^N_{i=1}\alpha_i\int_0^{\tau_i}e^{-rt} dD_i(t)\right),
\eeq
where the weight parameter satisfies $\alpha_1+\alpha_2+\ldots +\alpha_N=1$. The parameter $\alpha_i$ represents the relative weight of the subsidiary in the insurance group, and they add up to $1$ after scaling. $r>0$ is a given discount rate.  Recall that the insurance group manager is the decision maker, the surplus process of each subsidiary is therefore completely observable to the decision maker. The ruin time $\tau_i$ of the subsidiary $i$ is defined by
\begin{align*}
\quad\tau_i:=\inf\{t\geq0: X_i(t)=0\},\ \  i=1,\ldots, N.
\end{align*}
The initial surplus level is denoted by $X_i(0)=x_i$ and the initial default state is denoted by $Z_i(0)=z_i$, $i=1,\ldots, N$. We also denote
$\mathbf{X}(0)=\mathbf{x}:=(x_1,\ldots, x_N)$ and $\mathbf{Z}(0)=\mathbf{z}:=(z_1,\ldots,z_N)$. It is assumed henceforth that each admissible control process $D_i(t)$ can not jump simultaneously with $Z_i(t)$ in the sense that, for $t\geq 0$,
\begin{align}\label{simul_jump_assumption2}
  \Delta D_i(t)\Delta Z_i(t)=0, \ 1\leq i\leq N.
\end{align}
That is, the dividend for the subsidiary $i$ can not be paid right at the moment when the subsidiary $i$ goes default due to external credit risk. The assumption \eqref{simul_jump_assumption2} is by no means restrictive because the process $D_i(t)$ is c\`{a}dl\`{a}g and the default time $\sigma_i$ is totally inaccessible due to the existence of default intensity $\lambda_i$. In Appendix \ref{appx}, assumptions \eqref{simul_jump_assumption1} and \eqref{simul_jump_assumption2} are needed to derive the associated HJBVI. Moreover, it is assumed throughout the paper that $\Delta D_i(t)\leq X_i(t-)$ and $D_i(t)=D_i(t\wedge\tau_i)$, where the first condition dictates that the subsidiary $i$ can not pay dividend more than its currently available fund and the second condition means that the subsidiary $i$ won't pay any dividend after its ruin time.

Our goal is to find the optimal dividend strategy $\mathbf{D}^*$ such that the value function can be attained that
\begin{align}\label{valuedivd}
f(\mathbf{x},\mathbf{z}):=\sup_{D} J(\mathbf{x},\mathbf{z},\mathbf{D})=J(\mathbf{x},\mathbf{z},\mathbf{D}^*).
\end{align}
In particular, we are interested in the case that all subsidiaries are alive at the initial time, i.e., the value function $f(\mathbf{x},\mathbf{0})$ can be characterized, where $\mathbf{0}=(0,\ldots, 0)$ is the zero vector.

A barrier dividend strategy is to pay dividend whenever the surplus process excesses over the barrier. The optimal dividend for a single insurance company has been shown to fit this type of barrier control in various risk models. In our setting with default contagion, the optimal dividend for the insurance group also fits this barrier control. Nevertheless, the optimal barrier for each subsidiary is no longer a fixed level as in the model of a single insurance company. Instead, we identify that the optimal barrier is dynamically modulated by the defaulted subsidiaries and surviving ones. The dependence on the default state leads to some distinctive phenomena that the dividend barrier will be adjusted in the observation of sequential defaults. Furthermore, the change of the barrier for subsidiary $i$, i.e. the change of $m_i(\mathbf{Z}(t))$ in \eqref{optim divi}, is complicated and depends on all market parameters. In the case of two subsidiaries, we can prove in Corollary \ref{change-m} that the default event of one subsidiary will stimulate the surviving one to pay dividend, albeit with less amount, because the dividend threshold decreases.

For any vectors $\mathbf{x}\in[0,+\infty)^N$ and $\mathbf{z}\in\{0,1\}^N$, let us denote
\begin{align}
  \mathbf{x}^{(l)}:=(x_1,\ldots,x_{l-1},0,x_{l+1},\ldots,x_N),\quad\mathbf{z}^l:=(z_1,\ldots,z_{l-1},1,z_{l+1},\ldots,z_N).
\end{align}
\ \\
The next theorem is the main result of this paper.

\begin{thm}\label{mainthm1}
Let us consider the initial surplus level $\mathbf{X}(0)=\mathbf{x}\in[0,+\infty)^N$ and the initial default state $\mathbf{Z}(0)=\mathbf{z}:=(z_1,\ldots,z_N)=\mathbf{0}$ that all subsidiaries are alive at the initial time. The value function $f(\mathbf{x},\mathbf{0})$ defined in \eqref{valuedivd} is the unique classical solution to the variational inequalities
\begin{align}\label{varineqNc}
\max_{1\leq i\leq N}\left\{\mathcal{L}f(\mathbf{x},\mathbf{z})+\sum_{l=1}^N\lambda_l(\mathbf{z})f(\mathbf{x}^{(l)},\mathbf{z}^{l}),\alpha_i-\partial_if(\mathbf{x},\mathbf{z})\right\}=0,
\end{align}
in which the operator is defined by
\begin{align*}
\mathcal{L}f(\mathbf{x},\mathbf{z})&:=-\left(r+\sum_{k=1}^N\lambda_k(\mathbf{z})\right)f(\mathbf{x},\mathbf{z})+\sum_{k=1}^N\left(a_k\partial_k f(\mathbf{x},\mathbf{z})+\frac12b^2_k\partial_{kk} f(\mathbf{x},\mathbf{z})\right)+\sum_{\substack{i,j=1\\i>j}}^Nb_ib_j\rho_{ij}\partial^2_{ij}f(\mathbf{x},\mathbf{z}),
\end{align*}
where $\partial_k f:=\frac{\partial f}{\partial x_k}$ and $\partial_{kk} f:=\frac{\partial^2 f}{\partial x_k^2}$.

Moreover, for each $i=1,\ldots, N$, there exists a mapping $m_i:\{0,1\}^N \mapsto(0,+\infty)$ such that the optimal dividend $\mathbf{D}^*$ for the $i$-th subsidiary is given by the reflection strategy
\begin{align}\label{optim divi}
D^*_i(t):=\max\left\{0,\sup_{0\leq s\leq t}\left\{\wdt X_i(s)-m_i\left(\mathbf{Z}(s)\right)\right\}\right\},\quad i=1,\ldots,N,
\end{align}
and $m_i(\mathbf{Z}(t))$ represents the optimal barrier for the $i$-th subsidiary modulated by the N-dimensional default state indicator $\mathbf{Z}(t)$ at time $t$.
\end{thm}

From the form of HJBVI \eqref{varineqNc}, we can see that the solution $f(\mathbf{x},\mathbf{z})$ actually depends on the value function $f(\mathbf{x},\mathbf{z}^l)$ with the initial default state $\mathbf{z}^l$ indicating that one subsidiary has already defaulted. Therefore, to show the existence of classical solution to HJBVI \eqref{varineqNc} with $\mathbf{z}=\mathbf{0}$, we have to analyze the existence of the classical solution of the entire system of HJBVIs with all different values of $\mathbf{z}\in\{0,1\}^N$. To this end, we follow a recursive scheme that is based on default states of subsidiaries. The proof of Theorem \ref{mainthm1} is postponed to Section \ref{sec:verf}.

\section{Analysis of HJBVIs: Two Subsidiaries}\label{2companies}

To make our recursive arguments more readable, we first present the main result for only 2 subsidiaries. As one can see, the associated HJB variational inequalities can be solved explicitly for 2 initial subsidiaries and the optimal barriers of dividend for each subsidiary at time $t$ can be derived that depends on the default state $\mathbf{Z}(t)$. The recursive scheme to analyze the variational inequalities has a hierarchy feature, which is operated in a backward manner. To be more precise, we first solve a standard optimal dividend problem when only one subsidiary survives initially, and the associated value function appears as variable coefficients in the top level of HJBVI when both subsidiaries are initially alive. We can then continue to tackle the top level HJBVI with two subsidiaries by employing a separation form of its solution and the smooth-fit principle.

\subsection{\textit{One Surviving Subsidiary}}\label{3.1}
In this subsection, it is assumed that there is only one subsidiary at the initial time. That is, we need to consider default states ${\mathbf z_1}:=(0,1)$ and ${\mathbf z_2}:=(1,0)$. Here, the default state ${\mathbf z_i}$, $i=1,2$, indicates that subsidiary $i$ is alive initially while the other subsidiary has already defaulted due to the external credit risk.

For each $i$, let us consider the default state ${\mathbf z_i}$, and let $x_i\geq 0$ be the initial surplus level for the subsidiary $i$. The associated HJBVI for the default state $(0,1)$ and $(1,0)$ can be derived as
\begin{align}\label{f01}
\max\left\{\mathcal{L}^{{\mathbf z_i}}f(x_i,{\mathbf z_i}), \alpha_i-\frac{\partial f}{\partial x_i}(x_i,{\mathbf z_i})\right\}=0,\quad i=1,2,
\end{align}
where the operator is defined by
\begin{align*}
\mathcal{L}^{{\mathbf z_i}}f:=-\left(r+\lambda_i({\mathbf z_i})\right)f+\left(a_i\frac{\partial f}{\partial x_i}+\frac12b^2_i\frac{\partial^2f}{\partial x_i^2}\right).\notag
\end{align*}
Here, we recall that $\lambda_i({\mathbf z_i})$ stands for the default intensity for subsidiary $i$ given that the other subsidiary has already defaulted.

We can follow some standard results in \cite{AsmussenT}, which solves the stochastic control problem for a single insurance company. The positive discount rate $r>0$ ensures that $\frac12b^2_is^2+a_is-(r+\lambda_i({\mathbf z_i}))=0$ admits two real roots. Let $\hat{\theta}_{i1}$, $-\hat{\theta}_{i2}$ denote the positive and negative root respectively that
\begin{align*}
\hat{\theta}_{i1}:=\frac{-a_i + \sqrt{a^2_i+2b^2_i(r+\lambda_i({\mathbf z_i}))} }{b^2_i},\quad -\hat{\theta}_{i2}:=\frac{-a_i - \sqrt{a^2_i+2b^2_i(r+\lambda_i({\mathbf z_i}))} }{b^2_i},\quad i=1,2.
\end{align*}
According to results in \cite{AsmussenT}, for $i=1,2$, the solution to the HJBVI \eqref{f01} is
\begin{align}\label{func1}
f(x_i,{\mathbf z_i})=\left\{\begin{aligned}
\alpha_i C_i({\mathbf z_i})(e^{\hat{\theta}_{i1}x_i}-e^{-\hat{\theta}_{i2}x_i}),\quad &0\leq x_i\leq m_i({\mathbf z_i}),\\
\alpha_i C_i({\mathbf z_i})(e^{\hat{\theta}_{i1}m_i({\mathbf z_i})}-e^{-\hat{\theta}_{i2}m_i({\mathbf z_i})})+\alpha_i(x_i-m_i({\mathbf z_i})),\quad &x_i\geq m_{i}({\mathbf z_i}),
\end{aligned}\right.
\end{align}
where
\begin{align*}
m_i({\mathbf z_i})&:=\frac2{\hat{\theta}_{i1}+\hat{\theta}_{i2}}\log\left(\frac{\hat{\theta}_{i2}}{\hat{\theta}_{i1}}\right)=\frac{b^2_i}{\sqrt{a^2_i+2b^2_i(r+\lambda_i({\mathbf z_i}))}}\log\left(\frac{\sqrt{a^2_i+2b^2_i(r+\lambda_i({\mathbf z_i}))}+a_i}{\sqrt{a^2_i+2b^2_i(r+\lambda_i({\mathbf z_i}))}-a_i}\right),\notag\\
C_{i}({\mathbf z_i})&:=\frac1{\hat{\theta}_{i1}e^{\hat{\theta}_{i1}m_{i}({\mathbf z_i})}+\hat{\theta}_{i2}e^{-\hat{\theta}_{i2}m_{i}({\mathbf z_i})}},\quad i=1,2.
\end{align*}

\ \\
\subsection{\textit{Auxiliary Results for Two Subsidiaries}}

We continue to consider the case that both subsidiaries are alive at time $t=0$ with the initial surplus $\mathbf{x}=(x_1, x_2)$ and initial default state $\mathbf{z}=(0,0)$.
Using heuristic arguments in Appendix \ref{appx}, the associated HJBVI for the value function can be written by
\begin{align}\label{original HJB}
\max\left\{\mathcal{L}^{(0,0)}f(\mathbf{x},(0,0)),\alpha_1-\partial_1f(\mathbf{x},(0,0)),\alpha_2-\partial_2f(\mathbf{x},(0,0))\right\}=0,
\end{align}
with the operator
\begin{align}\label{L00}
\mathcal{L}^{(0,0)}f(\mathbf{x},(0,0)):=&-(r+\lambda_1(0,0)+\lambda_2(0,0)) f(\mathbf{x},(0,0))+b_1b_2\rho_{12}\partial_{12}f(\mathbf{x},(0,0))\notag\\
&+\left(a_1\partial_1 f(\mathbf{x},(0,0))+\frac12b^2_1\partial^2_{11} f(\mathbf{x},(0,0))\right)\notag\\
&+\left(a_2\partial_2 f(\mathbf{x},(0,0))+\frac12b^2_2\partial^2_{22} f(\mathbf{x},(0,0))\right)\notag\\
&+\lambda_1(0,0)f(x_2,(1,0))+\lambda_2(0,0)f(x_1,(0,1)),
\end{align}
where functions $f(x_1,(0,1))$ and $f(x_2,(1,0))$ are given explicitly in \eqref{func1}, and
\begin{align*}
\partial_if(\mathbf{x},(0,0)):=\frac{\partial f(\mathbf{x},(0,0))}{\partial x_i},\ \ \text{and}\ \ \partial_{ij}f(\mathbf{x},(0,0)):=\frac{\partial^2 f(\mathbf{x},(0,0))}{\partial x_ix_j},\ \ i,j=1,2.
\end{align*}
To show the existence of a classical solution to HJBVI \eqref{original HJB}, we first conjecture that the solution $f(\mathbf{x},(0,0))$ with $\mathbf{x}=(x_1, x_2)\in[0,+\infty)^2$ admits a key separation form that
\begin{align}\label{sepHJB}
f(\mathbf{x},(0,0))=f_1(x_1, (0,0))+f_2(x_2, (0,0)),\ \ \ x_1,x_2\geq 0,
\end{align}
for some smooth functions $f_1$ and $f_2$, i.e., functions of $x_1$ and $x_2$ can be decoupled. The rigorous proof of this separation form will be given in the next subsection.

With the aid of the separation form \eqref{sepHJB}, to solve HJBVI \eqref{original HJB} is equivalent to solve two auxiliary variational inequalities with one dimensional variable $x\in[0,+\infty)$ defined by
\begin{align}
\label{solu_1}\max\left\{\mathcal{A}_if_i(x, (0,0))+\frac{\lambda_1(0,0)\lambda_2(0,0)}{\lambda_i(0,0)}f(x,{\mathbf z_i}),\alpha_i-f'_i(x, (0,0))\right\}=0,\ \ i=1,2,\quad x\geq 0,
\end{align}
where the operators are defined as
\begin{align*}
\begin{aligned}
\mathcal{A}_if(x,(0,0)):=& \frac12b^2_if''(x,(0,0))+a_i f'(x,(0,0))-(r+\lambda_1(0,0)+\lambda_2(0,0))f(x,(0,0)),\quad i=1,2,
\end{aligned}
\end{align*}
and the boundary condition $f_i(0,(0,0))=0$, $i=1,2$.

\begin{rem}
When two subsidiaries are alive, the function $f_1(x_1, (0,0))$ from the decomposition relationship \eqref{decomppp} satisfies variational inequalities \eqref{solu_1}. It is worth noting that this function $f_1(x_1, (0,0))$ can not be simplify interpreted as the value function of the optimal dividend problem for the single subsidiary 1 without considering all other subsidiaries. As one can observe from \eqref{solu_1}, $f_1(x_1, (0,0))$ depends on the coefficient $\lambda_2(0,0)$ that is the default intensity of the subsidiary 2 and also depends on the value function $f_1(x, (0,1))$. However, as pointed out later in Remark \ref{rem5-1}, our mathematical approach can eventually verify that $f_1(x_1, (0,0))$ equals the expected value of the discounted dividend using the dividend control policy $D_1^*(t)$ for subsidiary 1, where $\mathbf{D}^*(t)=(D_1^*(t), D_2^*(t))$ is the optimal dividend for the whole group.
\end{rem}

By symmetry, for the existence of classical solution to the auxiliary variational inequality \eqref{solu_1}, for $i=1,2$, it is sufficient to study the general form of variational inequality with one dimensional variable $x\in[0,+\infty)$ defined by
\begin{align}\label{auxvainq-1}
\max\left\{\mathcal{A}f(x)+h(x), \gamma-f'(x)\right\}=0,
\end{align}
where $\gamma>0$,
\begin{align}\label{oprA}
\mathcal{A}f(x):=-\mu f(x)+\nu f'(x)+\frac12\sigma^2f''(x),\quad\mu,\nu,\sigma>0,
\end{align}
and the function $h$ is a $C^2$ function satisfying $h(0)=0$, $\lim_{u\rightarrow+\infty}h(u)=+\infty$, $h(x)\geq0$, $h'(x)>0$, and $h''(x)\leq0$, for $x\geq 0$.

To tackle the general variational inequality \eqref{auxvainq-1}, we propose to examine the solution to the ODE part at first in the next lemma.
\begin{lem}\label{0m part}
Let us consider the ODE problem
\begin{align}\label{HJBpde}
\mathcal{A}g(x)+h(x)=0,\ \ \ x\geq 0,
\end{align}
with the boundary condition $g(0)=0$ and the operator $\mathcal{A}$ is defined in \eqref{oprA}, $h$ is the same as that in \eqref{auxvainq-1}. The classical solution $g$ to \eqref{HJBpde} admits the form
\begin{align*}
g(x)=\phi_1(x)+C\phi_2(x),
\end{align*}
where $C$ is a parameter in $\mathbb{R}$, and
\begin{align}
\phi_1(x)&:=-\frac2{\sigma^2(\theta_1+\theta_2)}\int_0^xh(u)(e^{\theta_1(x-u)}-e^{-\theta_2(x-u)})du,\quad x\geq 0, \label{g_1}\\
\phi_2(x)&:=e^{\theta_1x}-e^{-\theta_2x},\quad x\geq 0.\label{g_2}\
\end{align}
Here $\theta_1$, $-\theta_2$ are the roots of the equation $\frac12\sigma^2\theta^2+\nu\theta-\mu=0$.
\end{lem}

\begin{proof} We first rewrite the ODE \eqref{HJBpde} in a vector form as
\begin{align*}
\frac d{dx}\left(
\begin{array}{c}
g(x)\\
g'(x)
\end{array}
\right)=A\left(
\begin{array}{c}
g(x)\\
g'(x)
\end{array}
\right)+\beta(x),
\end{align*}
where
\begin{align*}
A&:=\left(
\begin{array}{cc}
0&1\\
2\sigma^{-2}\mu&-2\sigma^{-2}\nu
\end{array}
\right),\ \ \ \ \
\beta(x):=\left(
\begin{array}{c}
0\\
-2\sigma^{-2}h(x)
\end{array}
\right).
\end{align*}
One can solve it as
\begin{align*}
\left(
\begin{array}{c}
g(x)\\
g'(x)
\end{array}
\right)=e^{Ax}\int_0^xe^{-Au}\beta(u)du+e^{Ax}\beta_0.
\end{align*}
The boundary condition $g(0)=0$ then yields that $\beta_0=(0,g'(0))^\top$ and
\begin{align*}
e^{Ax}\beta_0=\left(C(e^{\theta_1x}-e^{-\theta_2x}),C(\theta_1e^{\theta_1x}+\theta_2e^{-\theta_2x})\right)^\top,
\end{align*}
for some constant $C$. Note also that $\beta(x)=\big(0,-2\sigma^{-2}h(x)\big)$, hence it follows that
\begin{align*}
e^{Ax}\int_0^xe^{-Au}\beta(u)du&=-2\sigma^{-2}\int_0^xe^{A(x-u)}\left(
\begin{array}{c}
0\\
h(u)
\end{array}
\right)du\\
&=-2\sigma^{-2}\int_0^xh(u)e^{A(x-u)}\left(
\begin{array}{c}
0\\
1
\end{array}
\right)du.\notag
\end{align*}
Let
$\left(\begin{array}{c}
y_1(t)\\
y_2(t)
\end{array}\right)=e^{At}\left(\begin{array}{c}
0\\
1
\end{array}\right)$,
we get that
$\frac d{dt}\left(\begin{array}{c}
y_1(t)\\
y_2(t)
\end{array}\right)=A\left(\begin{array}{c}
y_1(t)\\
y_2(t)
\end{array}\right),
y_1(0)=0,y_2(0)=1$.
Then $y'_1(t)=y_2(t)$ implies that $y_1(t)=C_1e^{\theta_1t}+C_2e^{-\theta_2t}$, $y_1(0)=0$, $y'_1(0)=1$. We then deduce that $C_1=-C_2=\frac1{\theta_1+\theta_2}$.
Therefore, we have
\begin{align*}
e^{Ax}\int_0^xe^{-Au}\beta(u)du&=-2\sigma^{-2}\int_0^xh(u)e^{A(x-u)}\left(
\begin{array}{c}
0\\
1
\end{array}
\right)du\notag\\&=-\frac2{\sigma^2(\theta_1+\theta_2)}\int_0^x\left(
\begin{array}{c}
h(u)(e^{\theta_1(x-u)}-e^{-\theta_2(x-u)})\\
h(u)(\theta_1e^{\theta_1(x-u)}+\theta_2e^{-\theta_2(x-u)})
\end{array}
\right)du,
\end{align*}
and also
\begin{align*}
g(x,(0,0))&=-\frac2{\sigma^2(\theta_1+\theta_2)}\int_0^xh(u)(e^{\theta_1(x-u)}-e^{-\theta_2(x-u)})du+C(e^{\theta_1x}-e^{-\theta_2x})\notag\\
&=\phi_1(x)+C\phi_2(x),
\end{align*}
where $C$ is a parameter, and $\phi_1(x)$ and $\phi_2(x)$ satisfy \eqref{g_1} and \eqref{g_2} respectively.
\end{proof}

Back to the variational inequality \eqref{auxvainq-1}, we plan to apply the smooth-fit principle to mandate the solution to be smooth at the free boundary point. The next technical result becomes an important step to prove the main theorem.
\begin{lem}\label{HJBpara}
Under the conditions in Lemma \ref{0m part}, we have $\zeta>0$ and there exist positive constants $(C,m)$ such that
\begin{align*}
\left\{\begin{aligned}
\phi'_1(m)+C\phi'_2(m)=\gamma,\\
\phi''_1(m)+C\phi''_2(m)=0.
\end{aligned}\right.
\end{align*}
\end{lem}
\begin{proof}
Let us start with some identities of derivatives by direct calculations that
\begin{align}
\label{g'_1}\phi'_1(x)&=-\frac2{\sigma^2(\theta_1+\theta_2)}\int_0^xh(u)(\theta_1e^{\theta_1(x-u)}+\theta_2e^{-\theta_2(x-u)})du\leq0,\\
\label{g''_1}\phi''_1(x)&=-\frac2{\sigma^2(\theta_1+\theta_2)}\int_0^xh'(u)(\theta_1e^{\theta_1(x-u)}+\theta_2e^{-\theta_2(x-u)})du\leq0,
\end{align}
where the second inequality holds thanks to $h(0)=0$, and
\begin{align*}
\phi''_1(x)&=-\frac2{\sigma^2(\theta_1+\theta_2)}h(x)\phi'_2(0)-\frac2{\sigma^2(\theta_1+\theta_2)}\int_0^xh(u)\phi''_2(x-u)du\notag\\
&=-\frac2{\sigma^2(\theta_1+\theta_2)}\left(h(x)\phi'_2(0)-h(0)\phi'_2(x)\right)-\frac2{\sigma^2(\theta_1+\theta_2)}\int_0^xh(u)\phi''_2(x-u)du\notag\\
&=\frac2{\sigma^2(\theta_1+\theta_2)}\int_0^xh(u)\phi''_2(x-u)du-\frac2{\sigma^2(\theta_1+\theta_2)}\int_0^xh'(u)\phi'_2(x-u)du\notag\\
&\quad-\frac2{\sigma^2(\theta_1+\theta_2)}\int_0^xh(u)\phi''_2(x-u)du\notag\\
&=-\frac2{\sigma^2(\theta_1+\theta_2)}\int_0^xh'(u)\phi'_2(x-u)du.
\end{align*}
Note that $\phi''_2(0)=\theta^2_1-\theta^2_2<0$. As $\phi'_2(x)>0$, the existence of $m\in(0,+\infty)$ boils down to the existence of root $x\in(0,+\infty)$, to the following equation
\begin{align*}
q(x):=\phi''_1(x)+\frac{\gamma-\phi'_1(x)}{\phi'_2(x)}\phi''_2(x)=0.
\end{align*}
As $\phi'_1(0)=\phi''_1(0)=0$ by \eqref{g'_1} and \eqref{g''_1}, we obtain that $q(0)=\frac{\gamma\phi''_2(0)}{\phi'_2(0)}<0$.

Plugging \eqref{g'_1} and \eqref{g''_1} into the definition of $q$ above, we obtain that
\begin{align}\label{q-expression}
  q(x)=&\gamma\frac{\phi''_2(x)}{\phi'_2(x)}+\frac2{\sigma^2(\theta_1+\theta_2)}\int_0^x\left[\frac{\phi''_2(x)}{\phi'_2(x)}h(u)-h'(u)\right](\theta_1e^{\theta_1(x-u)}+\theta_2e^{-\theta_2(x-u)})du.
\end{align}
As $h''\leq0$, $h'>0$, it follows that $h'$ is bounded. Noting that $\lim_{x\rightarrow+\infty}\frac{\phi''_2(x)}{\phi'_2(x)}=\theta_1>0$,
as well as that $\lim_{u\rightarrow+\infty}h(u)=+\infty$, we deduce from \eqref{q-expression} that $\lim_{x\rightarrow+\infty}q(x)=+\infty$. Therefore $q$ admits at least one root $x\in(0,+\infty)$. We then define
\begin{align}\label{nowdefm}
m:=\inf\left\{u:q(u)=0\right\}\in(0,+\infty),
\end{align}
and choose
\begin{align}\label{nowdefC}
C:=\frac{\gamma-\phi'_1(m)}{\phi'_2(m)}\geq\frac{\gamma}{\phi'_2(m)}>0.
\end{align}
\end{proof}

With the parameters $(C,m)$ obtained in \eqref{nowdefC} and \eqref{nowdefm} in the proof of Lemma \ref{HJBpara}, we can turn to the construction of a classical solution to the general variational inequality.

\begin{prop}\label{THHJB}
The variational inequality
\begin{align}\label{abs_aux}
\max\left\{\mathcal{A}f(x)+h(x),\gamma-f'(x)\right\}=0,\ \ \ x\geq 0,
\end{align}
with the boundary condition $f(0)=0$ admits a $C^2$ solution, which has the form of
\begin{align}\label{fzi expression}
f(x)=\left\{\begin{aligned}
\phi_1(x)+C\phi_2(x),\quad &x\in[0,m],\\
\phi_1(m)+C\phi_2(m)+\gamma(x-m),\quad &x\in[m,+\infty).
\end{aligned}\right.
\end{align}
Here $\phi_1(x)$ and $\phi_2(x)$, $x\geq 0$, are defined in \eqref{g_1} and \eqref{g_2} respectively and parameters $C$ and $m$ are determined in \eqref{nowdefC} and \eqref{nowdefm}.

In particular, we have
\begin{align}\label{HJBVI-specific}
\left\{\begin{aligned}
\mathcal{A}f(x)+h(x)=0,\quad &x\in[0,m],\\
\gamma-f'(x)=0,\quad &x\in[m,+\infty),
\end{aligned}\right.
\end{align}
and $f(0)=0$, $f'>0$, $f''\leq0$, $\lim_{x\rightarrow+\infty}f(x)=+\infty$.
\end{prop}

\begin{proof}[Proof of Proposition \ref{THHJB}]
Let $g(x)$ be the classical solution to the ODE \eqref{HJBpde}. We have that $f(x)$ coincides with $g(x)$ in Lemma \ref{0m part}, for $x\leq m$ and the function is a linear function, for $x>m$. We aim to prove that the function $f$ is the desired $C^2$ solution to the variational inequality \eqref{abs_aux}. Thanks to Lemma \ref{HJBpara}, we deduce that $f'(m)=\gamma$, $f''(m)=0$. In view of its definition, it is straightforward to see that $f$ belongs to $C^2$. On the other hand, Lemma \ref{0m part} and \eqref{fzi expression} give the validity of \eqref{HJBVI-specific}. Therefore \eqref{abs_aux} holds once we show that
\begin{align*}
f'(x)=\phi'_1(x)+C\phi'_2(x)\geq \gamma,\ \ \text{for}\ x\in[0,m],
\end{align*}
as well as
\begin{align*}
\mathcal{A}f(x)+h(x)\leq 0,\ \ \text{for}\ x\geq m.
\end{align*}
Define the elliptic operator
\begin{align}
  Lf:=-\frac12\sigma^2f''-\nu f'+\mu f,
\end{align}
and consider $g(x)$ in Lemma \ref{0m part} with $C$ in \eqref{nowdefC}. Then we have
\begin{align*}
  Lg(x)=h(x),\quad x\in(0,m).
\end{align*}
Note that $h$ is twice differentiable, and that $h''\leq0$. It therefore follows that
\begin{align*}
  Lg''(x)=h''(x)\leq0,\quad x\in(0,m).
\end{align*}
Since $\mu>0$, according to the weak maximum principle (see Theorem 2 in $\S6.4$ of \cite{EvansPDE}), we have
\begin{align*}
  \max_{x\in[0,m]}g''(x)\leq\max\left\{\big[g''(0)\big]^+,\big[g''(m)\big]^+\right\}=0.
\end{align*}
Therefore, we have
\begin{align*}
 \phi'_1(x)+C\phi'_2(x)\geq \phi'_1(m)+C\phi'_2(m)=\gamma,\ \ \text{for}\ x\in[0,m].
\end{align*}
In other words,
\begin{align}\label{f''-1}
  \phi''_1(x)+C\phi''_2(x)\leq0,\quad x\in[0,m].
\end{align}

We next show that $\mathcal{A}f'(x)+h'(x)\leq0$, for $x\geq m$. In our previous argument, we have shown that $\phi''_1(x)+C\phi''_2(x)\leq0$, $x\in[0,m]$, i.e.,  $f''(x)\leq0$, $x\in[0,m]$. It follows that
\begin{align}\label{f'''}
f'''(m-)=\lim_{x\rightarrow m-}\frac{f''(m)-f''(x)}{m-x}=-\lim_{x\rightarrow m-}\frac{f''(x)}{m-x}\geq0.
\end{align}
Thanks to the definition of $f$, we have that $\mathcal{A}f'(x)+h'(x)=0$ on $x\in[0,m)$. By sending $x\rightarrow m-$, we get
\begin{align*}
\mathcal{A}f'(m-)+h'(m)=0.
\end{align*}
That is,
\begin{align*}
-\mu \gamma+h'(m)=-\frac12\sigma^2f'''(m-)\leq0.
\end{align*}
For $x>m$, we have $f''(x)=0$, $f'(x)=\gamma$, and $h'(x)\leq h'(m)$ as $h''\leq0$. Hence, we have
\begin{align*}
\mathcal{A}f'(x)+h'(x)=-\mu f'(x)+h'(x)\leq-\mu \gamma+h'(m)\leq0.
\end{align*}
Then for $x\geq m$, we arrive at
\begin{align*}
\mathcal{A}f(x)+h(x)\leq\mathcal{A}f(m)+h(m)=0.
\end{align*}
Putting all the pieces together, we can conclude that $f$ is the desired $C^2$ solution to the variational inequality \eqref{abs_aux}.

To complete the proof, it remains to show that
\begin{align*}
f(0)=0,\ \ f'(x)>0,\ \ f''(x)\leq0,\ \ x\geq0.
\end{align*}
In view of \eqref{g_1}, \eqref{g_2} and \eqref{fzi expression}, it holds that $f(0)=0$. Note that the variational inequality \eqref{abs_aux} gives $f'(x)>0$, $x\geq0$. Moreover, in view of \eqref{f''-1} and the fact that $f(x)$ is linear on $x\in[m,+\infty)$, we obtain that $f''(x)\leq0$, $x\geq0$, $\lim_{x\rightarrow+\infty}f(x)=+\infty$.
\end{proof}

\subsection{\textit{Main Results for Two Subsidiaries}}
In view of the explicit solution of the auxiliary variational inequality \eqref{abs_aux}, for $i=1,2$, we can derive the explicit solution $f_i(x_i,(0,0))$ to the variational inequality \eqref{solu_1} by setting $\mathcal{A}=\mathcal{A}_i$, $h(x_i) = \frac{\lambda_1(0,0)\lambda_2(0,0)}{\lambda_i(0,0)}f_i(x_i,{\mathbf z_i})$ and $\gamma=\alpha_i$.

Moreover, for $i=1,2$, let us denote the constant $m$ and $C$ for variational inequality \eqref{solu_1} by $m_i(0,0)$ and $C_i(0,0)$, because we can verify later that the constant $m_i(0,0)$ is the optimal barrier of the dividend strategy for the subsidiary $i$.

Let us define $K_i :=\alpha_i C_i({\mathbf z_i})(e^{\hat{\theta}_{i1}m_i({\mathbf z_i})}-e^{-\hat{\theta}_{i2}m_i({\mathbf z_i})})-\alpha_im_i({\mathbf z_i})$, $i=1,2$, and we will construct the explicit solution of the variational inequality (\ref{solu_1}) in the following steps.

For $i=1,2$, let us denote ${\theta}_{i1}$, $-{\theta}_{i2}$ as the positive and negative roots of the equation $\frac12b^2_i{\theta}^2+a_i{\theta}-(r+\lambda_1(0,0)+\lambda_2(0,0))=0$ respectively that
\begin{align*}
{\theta}_{i1} &:=\frac{-a_i + \sqrt{a^2_i+2b^2_i(r+\lambda_1(0,0)+\lambda_2(0,0))} }{b^2_i},\\
\ \\
-{\theta}_{i2} &:=\frac{-a_i - \sqrt{a^2_i+2b^2_i(r+\lambda_1(0,0)+\lambda_2(0,0))} }{b^2_i}.
\end{align*}
Let us first define for $i=1,2$ and the variable $x\geq 0$ that
{\small
\begin{align}\label{f11}
f_{i1}(x,(0,0))
&:=\left\{\begin{aligned}
&f_{i11}(x):=-\frac2{\sigma^2}\frac{\alpha_i \lambda_1(0,0)\lambda_2(0,0)C_i({\mathbf z_i})}{\lambda_i(0,0)(\theta_{i1}+\theta_{i2})}\\
& \times
\Bigg[ \frac{(\theta_{i1}+\theta_{i2})e^{\hat\theta_{i1}x}}{(\hat\theta_{i1}-\theta_{i1})(\hat\theta_{i1}+\theta_{i2})}+ \frac{(\theta_{i1}+\theta_{i2})e^{-\hat\theta_{i2}x}}{(\hat\theta_{i2}+\theta_{i1})(-\hat\theta_{i2}+\theta_{i2})}\\
& -\frac{(\hat\theta_{i1}+\hat{\theta}_{i2})e^{\theta_{i1}x}}{(\hat\theta_{i1}-\theta_{i1})(\hat\theta_{i2}+\theta_{i1})}-\frac{(\hat\theta_{i1}+\hat{\theta}_{i2}) e^{-\theta_{i2}x}}{(\hat\theta_{i1}+\theta_{i2})(-\hat\theta_{i2}+\theta_{i2})}
\Bigg],\quad0\leq x\leq m_i({\mathbf z_i}),\\
&f_{i12}(x):=-\frac2{\sigma^2}\frac{\alpha_i \lambda_1(0,0)\lambda_2(0,0)C_i({\mathbf z_i})}{\lambda_i(0,0)(\theta_{i1}+\theta_{i2})}\\
& \times
\Bigg[  \frac{e^{\theta_{i1}x} }{\hat\theta_{i1}-\theta_{i1}}\Big(e^{(\hat\theta_{i1}-\theta_{i1})m_i({\mathbf z_i})}-1\Big)+ \frac{ e^{-\theta_{i2}x}}{\hat\theta_{i1}+\theta_{i2}}\Big(-e^{(\hat\theta_{i1}+\theta_{i2})m_i({\mathbf z_i})}+1\Big)\\
&+ \frac{e^{\theta_{i1}x} }{\hat\theta_{i2}+\theta_{i1}}\Big(e^{-(\hat\theta_{i2}+\theta_{i1})m_i({\mathbf z_i})}-1\Big)+ \frac{ e^{-\theta_{i2}x}}{-\hat\theta_{i2}+\theta_{i2}}\Big(e^{(-\hat\theta_{i2}+\theta_{i2})m_i({\mathbf z_i})}-1\Big)
\Bigg]\\
&-\frac2{\sigma^2}\frac{K_i\lambda_1(0,0)\lambda_2(0,0)}{\lambda_i(0,0)(\theta_{i1}+\theta_{i2})}\\
&\times\Bigg[
\frac{1}{\theta_{i1}}\Big(e^{\theta_{i1}x-\theta_{i1}m_i({\mathbf z_i})}-1\Big) + \frac{1}{\theta_{i2}}\Big(e^{-\theta_{i2}x+\theta_{i2}m_i({\mathbf z_i})}-1\Big)
\Bigg]\\
&-\frac2{\sigma^2}\frac{\alpha_i\lambda_1(0,0)\lambda_2(0,0)}{\lambda_i(0,0)(\theta_{i1}+\theta_{i2})}\\
&\times\Bigg[
\frac{1}{(\theta_{i1})^2}\Big(-\theta_{i1}x-1+(\theta_{i1}m_i({\mathbf z_i})+1)e^{\theta_{i1}x-\theta_{i1}m_i({\mathbf z_i})}\Big)\\
&+ \frac{1}{(\theta_{i2})^2}\Big(-\theta_{i2}x+1+(\theta_{i2}m_i({\mathbf z_i})-1)e^{-\theta_{i2}x+\theta_{i2}m_i({\mathbf z_i})}\Big)
\Bigg],\quad m_i({\mathbf z_i})\leq x,
\end{aligned}\right.\\
f_{i2}(x,(0,0))&=e^{\theta_{i1}x}-e^{-\theta_{i2}x},\ \ x\geq 0.\label{f12}
\end{align}}
In view of Lemma \ref{HJBpara} and Proposition \ref{THHJB}, we can define the constant
\begin{align*}
m_i(0,0):=\inf\{s:q_i(s)=0\},\quad i=1,2,
\end{align*}
where
\begin{align*}
q_i(x):=f''_{i1}(x,(0,0))+\frac{\alpha_i-f'_{i1}(x,(0,0))}{f'_{i2}(x,(0,0))}f''_{i2}(x,(0,0)),\quad i=1,2.
\end{align*}
We also define $C_i(0,0):=\frac{\alpha_i-f'_{i1}(m_i(0,0))}{f'_{i2}(m_i(0,0))}$, $i=1,2$.

To illustrate the change of the optimal barrier when one subsidiary defaults, let us choose the model parameters:  $a_1=0.1$, $b_1=0.07$, $a_2=0.15$, $b_2=0.06$, $\lambda_1(0,0)=0.02$, $\lambda_1(0,1)=0.04$, $\lambda_2(0,0)=0.01$, $\lambda_2(1,0)=0.04$, $r=0.05$ and $\alpha_1=0.4$. We can see from Figure $1$ that the comparison results $m_1(0,0)>m_1(0,1)$ and $m_2(0,0)>m_2(1,0)$ hold. That is, both subsidiaries decrease the optimal barriers for dividend payment after the other subsidiary defaults. These observations are consistent with our intuition that the default contagion effect forces the surviving subsidiary to take into account that itself will go default very soon because of the increased default intensity. Therefore the surviving one prefers to pay dividend as soon as possible by setting a lower dividend threshold before the unexpected default happens.

$$
\begin{array}{ccc}
\begin{array}{c}
\includegraphics[height=2.6in]{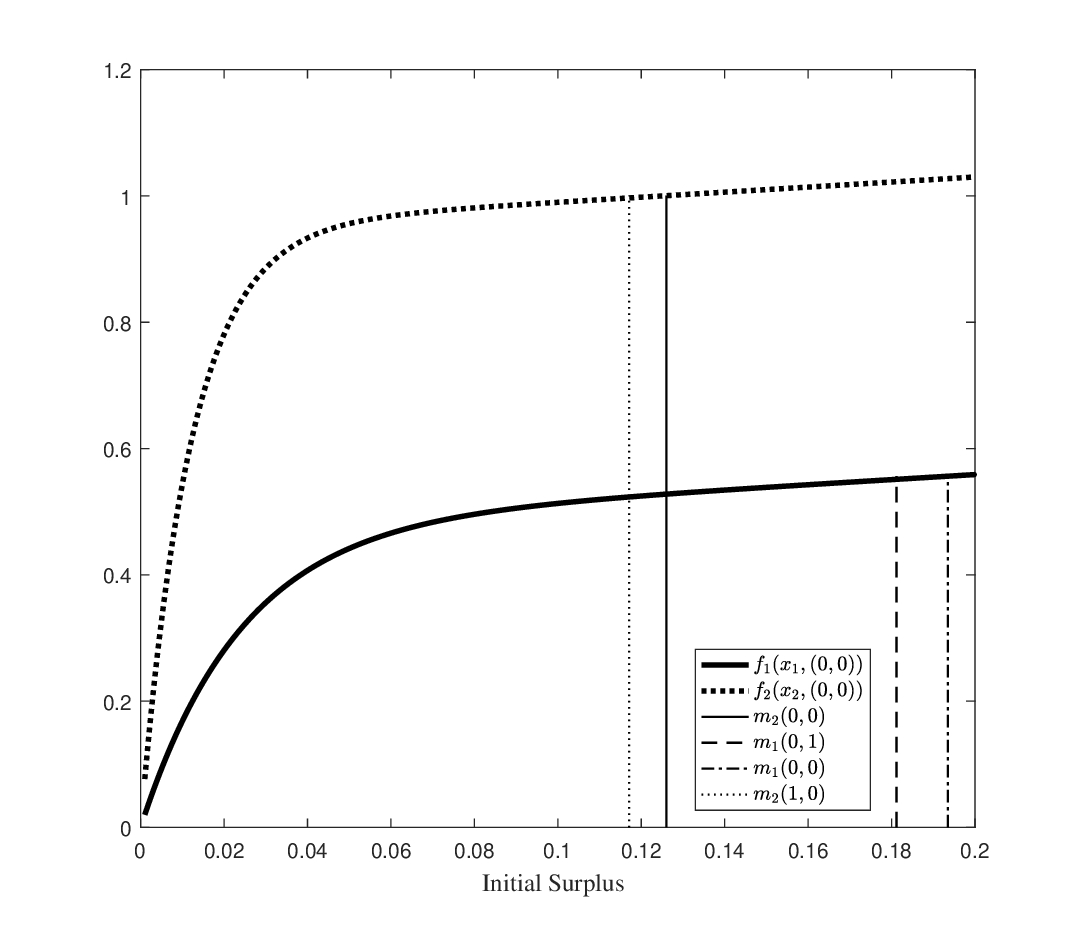} \\
\mbox{Figure 1: {\small The change of the optimal barrier when default occurs}}
\end{array}
\end{array}
$$
\ \\
We actually have the next theoretical result on the change of the optimal barrier when one subsidiary defaults.

\begin{cor}\label{change-m}
For the case of two subsidiaries, as we have $\lambda_1(0,1)\geq\lambda_1(0,0)$ and $\lambda_2(1,0)\geq\lambda_2(0,0)$, we always have the orders that $m_1(0,0)\geq m_1(0,1)$ and $m_2(0,0)\geq m_2(1,0)$.
\end{cor}

\begin{proof}
It suffices to show that $m_1(0,0)\geq m_1(0,1)$. We first show that $f_1(x,(0,0))\geq f_1(x,(0,1))$, $x\geq0$. Define $f_\delta(x):=e^{-\delta x}f_1(x,(0,0))$, $\hat{f}_\delta(x):=e^{-\delta x}f_1(x,(0,1))$. Here, we choose the constant $\delta>0$ small enough such that $r+\lambda_1(0,0)+\lambda_2(0,0)-\delta a_1-\frac12\delta^2b^2_1>0$. We can verify by direct calculation that $f_\delta(x)$ satisfies
  \begin{align}\label{}
  \max\left\{\mathcal{A}^\delta_1f_\delta(x)+\lambda_2(0,0)f_1(x,(0,1)),\alpha_1-\big(e^{\delta x}f_\delta(x)\big)'\right\}=0,\quad x\geq 0,
  \end{align}
  with $f_\delta(0)=0$ and the operator $\mathcal{A}^\delta_1$ defined by
  \begin{align*}
    \mathcal{A}^\delta_1f:=\frac12b^2_1\big(e^{\delta x}f(x)\big)''+a_1\big(e^{\delta x}f(x)\big)'-\big(r+\lambda_1(0,0)+\lambda_2(0,0)\big)e^{\delta x}f(x).
  \end{align*}
  On the other hand, we have that
  \begin{align*}\label{}
  \max\left\{\hat{\mathcal{A}}^\delta_1\hat{f}_\delta(x)+\lambda_2(0,0)f_1(x,(0,1)),\alpha_1-\big(e^{\delta x}\hat{f}_\delta(x)\big)'\right\}=0,\quad x\geq 0,
  \end{align*}
  with $\hat{f}_\delta(0)=0$ and the operator $ \hat{\mathcal{A}}^\delta_1$ defined by
  \begin{align*}
    \hat{\mathcal{A}}^\delta_1f:=\frac12b^2_1\big(e^{\delta x}f(x)\big)''+a_1\big(e^{\delta x}f(x)\big)'-\big(r+\hat{\lambda}_1(0,0)+\lambda_2(0,0)\big)e^{\delta x}f(x),
  \end{align*}
  and $\hat{\lambda}_1(0,0):=\lambda_1(0,1)$. Noting that $\hat{\lambda}_1(0,0)\geq\lambda_1(0,0)$ and $\hat{f}_\delta\geq0$, we thus have that
  \begin{align*}
  \max\left\{\mathcal{A}^\delta_1\hat{f}_\delta(x)+\lambda_2(0,0)f_1(x,(0,1)),\alpha_1- \big(e^{\delta x}\hat{f}_\delta(x)\big)' \right\}\geq0,\quad x\geq 0.
  \end{align*}
  The comparison result of viscosity solutions (see e.g. Section 5B in \cite{viscosity-guide}) yields that, for each $M>0$,
  \begin{align*}
    \hat{f}_\delta(x)-f_\delta(x)\leq\max\left\{0,\hat{f}_\delta(M)-f_\delta(M)\right\},\quad x\in[0,M].
  \end{align*}
Note that $M>0$ is arbitrary and $\lim_{M\rightarrow+\infty}|\hat{f}_\delta(M)-f_\delta(M)|=0$. Letting $M\rightarrow+\infty$ in the inequality above, we obtain that
  \begin{align*}
    f_\delta(x)-\hat{f}_\delta(x)\geq0,\quad x\geq0.
  \end{align*}
This gives that $f_1(x,(0,0))\geq f_1(x,(0,1))$, $x\geq0$.

Next, let us define $g(x_1):=f'_1(x_1,(0,0))$. We claim that $g$ is the viscosity solution of
\begin{align}\label{viscosity-g}
  \max\left\{\mathcal{A}_1g(x_1)+\lambda_2(0,0)f'_1(x_1,(0,1)),\alpha_1-g(x_1)\right\}=0,
\end{align}
with $g(0)=f'_1(0,(0,0))$ and $g(M)=\alpha_1$, where the constant $M$ is sufficiently large that $M>m_1(0,1)\vee m_1(0,0)$. Indeed, on $(0,+\infty)\setminus\{m_1(0,1)\}$, $g$ is $C^2$ and satisfies \eqref{viscosity-g}. On the other hand, similar to \eqref{f'''}, we can derive that
\begin{align*}
  \lim_{x\uparrow m_1(0,1)}g''(x)=\lim_{x\uparrow m_1(0,1)}f'''_1(x,(0,1))\geq0,
\end{align*}
as well as that $\lim_{x\downarrow m_1(0,1)}g''(x)=0$. Hence
\begin{align*}
  &D^{+(2)}g(m_1(0,1))=\left\{(0,p):p\geq\lim_{x\uparrow m_1(0,1)}f'''_1(x,(0,1))\right\},\\
  &D^{-(2)}g(m_1(0,1))=\left\{(0,p):p\leq0\right\}.
\end{align*}
Here, we denote $D^{+(2)}$ and $D^{-(2)}$ the second order Super-Jet and Sub-Jet respectively. For $(0,p)\in D^{+(2)}g(m_1(0,1))$, we have that
\begin{align*}
  \max\left\{\frac12b^2_1\cdot p+a_1\cdot0-(r+\lambda_1(0,0)+\lambda_2(0,0))g(m_1(0,1)),\alpha_1-g(m_1(0,1))\right\}\geq0,
\end{align*}
while for $(0,p)\in D^{-(2)}g(m_1(0,1))$, we have
\begin{align*}
  \max\left\{\frac12b^2_1\cdot p+a_1\cdot0-(r+\lambda_1(0,0)+\lambda_2(0,0))g(m_1(0,1)),\alpha_1-g(m_1(0,1))\right\}\leq0.
\end{align*}
Therefore $g$ is the viscosity solution of \eqref{viscosity-g}.

Let us define $\hat{g}(x):=f'_1(x,(0,1))$. Following the same arguments above, we have that $\hat{g}$ is the viscosity supersolution to \eqref{viscosity-g}, or equivalently, the viscosity solution to
\begin{align}
  \max\left\{\mathcal{A}_1\hat{g}(x_1)+\lambda_2(0,0)f'_1(x_1,(0,1)),\alpha_1-\hat{g}(x_1)\right\}\geq0,
\end{align}
with $\hat{g}(0)=f'_1(0,(0,1))$ and $\hat{g}(M)=\alpha_1$.

Because we have shown that
\begin{align*}
  f_1(x,(0,0))\geq f_1(x,(0,1)),\quad f_1(0,(0,0))=f_1(0,(0,1))=0,
\end{align*}
it follows that $f'_1(0,(0,0))\geq f'_1(0,(0,1))$, i.e., $g(0)\geq\hat{g}(0)$. Moreover, $g(M)=\hat{g}(M)=\alpha_1$. The comparison result of viscosity solutions gives that $g(x)\geq\hat{g}(x)$, $x\in[0,M]$. That is, $f'_1(x,(0,0))\geq f'_1(x,(0,1))$. We thus deduce that
\begin{align*}
  \alpha_1=f'_1(m_1(0,0),(0,0))\geq f'_1(m_1(0,0),(0,1))\geq\alpha_1,
\end{align*}
which implies that $f'_1(m_1(0,0),(0,1))=\alpha_1$. As $f'_1(x,(0,1))>\alpha_1$, for $x\in(0,m_1(0,1))$, we can obtain the desired order that $m_1(0,1)\leq m_1(0,0)$.
\end{proof}

\ \\
Based on solution forms in \eqref{f11} and \eqref{f12} and Corollary \ref{change-m}, we have $m_i(0,0) \geq m_i({\mathbf z_i})$, $i=1,2$, and the solution of the auxiliary variational inequality (\ref{solu_1}) satisfies the piecewise form that
\begin{align}\label{pieces12}
f_i(x_i,(0,0))=\left\{\begin{aligned}
f_{i11}(x_i)+C_i(0,0)f_{i2}(x_i,(0,0))&,\quad\quad &0\leq x_i< m_i({\mathbf z_i}),\\
f_{i12}(x_i)+C_i(0,0)f_{i2}(x_i,(0,0))&,\quad\quad &m_i({\mathbf z_i})\leq x_i\leq m_i(0,0),\\
f_{i12}(m_i(0,0))+C_i(0,0)f_{i2}(m_i(0,0),(0,0))&\\
+\alpha_i(x_i-m_i(0,0))&,\quad\quad &x_i>m_i(0,0).
\end{aligned}\right.
\end{align}

We can continue to verify the important conjecture $f(\mathbf{x},(0,0))=f_1(x_1, (0,0))+f_2(x_2, (0,0))$ in \eqref{sepHJB} and prove the existence of a classical solution to HJBVI \eqref{original HJB} in the next theorem.

\begin{thm}\label{soltwo}
There exists a $C^2$ solution to HJBVI \eqref{original HJB} that admits the form
\begin{align}\label{decomppp}
  f(\mathbf{x},(0,0)):=f_1(x_1,(0,0))+f_2(x_2,(0,0)),
\end{align}
where $f_i(x_i, (0,0))$ given in \eqref{pieces12} is the $C^2$ solution to the auxiliary variational inequality \eqref{solu_1}, $i=1,2$.
\end{thm}
\begin{proof}
Thanks to Proposition \ref{THHJB}, the auxiliary variational inequality \eqref{solu_1} admits $C^2$ solution, for $i=1,2$. Let $f_i$ be the solution to \eqref{solu_1}, $i=1,2$. By setting $f(\mathbf{x},(0,0)):=f_1(x_1,(0,0))+f_2(x_2,(0,0))$ and plugging into \eqref{L00}, we have
\begin{align*}
\mathcal{L}^{(0,0)}f(\mathbf{x},(0,0))=&-rf_1(x_1,(0,0))-rf_2(x_2,(0,0))\notag\\
&+\left(a_1\partial_1 f_1(x_1,(0,0))+\frac12b^2_1\partial^2_{11} f_1(x_1,(0,0))\right)\notag\\
&-\big(\lambda_1(0,0)+\lambda_2(0,0)\big)f_1(x_1,(0,0))+\lambda_2(0,0)f(x_1,(0,1))\notag\\
&+\left(a_2\partial_2 f_2(x_2,(0,0))+\frac12b^2_2\partial^2_{22} f_2(x_2,(0,0))\right)\notag\\
&-\big(\lambda_1(0,0)+\lambda_2(0,0)\big)f_2(x_2,(0,0))+\lambda_2(0,0)f(x_1,(0,1)).\notag
\end{align*}
It readily yields that
\begin{align*}
&\mathcal{L}^{(0,0)}f(\mathbf{x},(0,0))=\mathcal{A}_1f_1(x_1,(0,0))+\lambda_2(0,0)f_1(x_1,(0,1))+\mathcal{A}_2f_2(x_2,(0,0))+\lambda_1(0,0)f_2(x_2,(1,0)),\notag\\
&\alpha_1-\partial_1f(\mathbf{x},(0,0))=\alpha_1-f'_1(x_1,(0,0)),\notag\\
&\alpha_2-\partial_2f(\mathbf{x},(0,0))=\alpha_2-f'_2(x_2,(0,0)).
\end{align*}
As $f_i$ solves the variational inequality \eqref{solu_1}, $i=1,2$, we have that
\begin{align*}
\max\left\{\mathcal{L}^{(0,0)}f(\mathbf{x},(0,0)),\alpha_1-\partial_1f(\mathbf{x},(0,0)),\alpha_2-\partial_2f(\mathbf{x},(0,0))\right\}\leq0.
\end{align*}
Moreover, if $\mathcal{L}^{(0,0)}f(\mathbf{x},(0,0))<0$, we get that $$\mathcal{A}_1f_1(x_1,(0,0))+\lambda_2(0,0)f(x_1,(0,1))<0 \quad\text{or}\quad \mathcal{A}_2f_2(x_2,(0,0))+\lambda_1(0,0)f(x_2,(1,0))<0.$$ Without loss of generality, we assume that $\mathcal{A}_1f_1(x_1,(0,0))+\lambda_2(0,0)f(x_1,(0,1))<0$. By \eqref{solu_1}, we have that $\alpha-\partial_1f(\mathbf{x},(0,0))=\alpha-f'_1(x_1,(0,0))=0$, and hence
\begin{align*}
\max\left\{\mathcal{L}^{(0,0)}f(\mathbf{x},(0,0)),\alpha-\partial_1f(\mathbf{x},(0,0)),1-\alpha-\partial_2f(\mathbf{x},(0,0))\right\}=0.
\end{align*}
This shows that $f(\mathbf{x},(0,0))$ in \eqref{decomppp} is the solution of the HJBVI \eqref{original HJB}.
\end{proof}

\section{Analysis of HJBVIs: Multiple Subsidiaries}\label{Multi comp}
This section generalizes the previous results to the case with $N\geq 3$ subsidiaries by employing mathematical induction. To this end, let us start to focus on the case that there are $k\leq N$ subsidiaries defaulted at the initial time and show the existence of classical solution to the associated variational inequality. The final verification proof of the optimal reflection dividend strategy for $N$ initial subsidiaries is given in the next section.

For $0\leq k\leq N$, let us consider the initial default state that $k$ subsidiaries have defaulted and denote $\mathbf{z}=0^{j_1,\ldots,j_k}$ as the $N$ dimensional vector that $j_1$, $\ldots$, $j_k$ components are $1$ and all other components are $0$ if $k\geq 1$ and denote $\mathbf{z}=0^{j_1,\ldots,j_k}$ as the N-dimensional zero vector $\mathbf{0}$ if $k=0$. We also denote by $\left\{j_{k+1},\ldots,j_N\right\}:=\left\{1,2,\ldots,N\right\}\backslash\left\{j_1,\ldots,j_k\right\}$. For example, if $(j_1,\ldots,j_k)=(1,2,\ldots,k)$, then $(j_{k+1},\cdots,j_{N})=(k+1,\ldots,N)$.

Consider $\mathbf{z}=0^{1,\ldots,k}$, $\mathbf{x}=(0,\ldots,0,x_{k+1},\ldots,x_N)$, and define the operator
\begin{align}\label{L^z}
\mathcal{L}^{\mathbf{z}}f(\mathbf{x},\mathbf{z}):=&-\left(r+\sum_{i=k+1}^N\lambda_{i}(\mathbf{z})\right)f(\mathbf{x},\mathbf{z})+\sum_{i=k+1}^N\left(a_i\partial_i f(\mathbf{x},\mathbf{z})+\frac12b^2_i\partial^2_{ii} f(\mathbf{x},\mathbf{z})\right)\\
&+\sum_{\substack{i,l=k+1\\i<l}}^Nb_ib_l\rho_{il}\partial^2_{il}f(\mathbf{x},\mathbf{z}).\notag
\end{align}
With the notation above, we introduce the recursive system of HJBVIs
\begin{align}\label{}
\max_{k+1\leq i\leq N}\left\{\mathcal{L}^{\mathbf{z}}f(\mathbf{x},\mathbf{z})+\sum_{l=k+1}^N\lambda_l(\mathbf{z})f(\mathbf{x}^{(l)},\mathbf{z}^l),\alpha_i-\partial_if(\mathbf{x},\mathbf{z})\right\}=0.
\end{align}
Similar to the previous section, we seek for the solution in the separation form $$f(\mathbf{x},\mathbf{z})=\sum_{i=k+1}^Nf_{i}(x_{i} ,\mathbf{z}),$$ so that $x_{k+1}$, $\ldots$, $x_N$ are decoupled, where we define, for any $x\geq 0$, that
\begin{align}\label{fzi expression0}
f_i(x,\mathbf{z})=\left\{\begin{aligned}
f_{i,1}(x,\mathbf{z})+C_i(\mathbf{z})f_{i,2}(x,\mathbf{z}),\quad &0\leq x\leq m_i(\mathbf{z}),\\
f_{i,1}(m_i(\mathbf{z}),\mathbf{z})+C_i(\mathbf{z})f_{i,2}(m_i(\mathbf{z}),\mathbf{z})+\alpha_i(x-m_i(\mathbf{z})),\quad &x\geq m_i(\mathbf{z}).
\end{aligned}\right.
\end{align}
In particular, for $k+1\leq i\leq N$,
\begin{align}\label{U-i-and-U}
\left\{\begin{aligned}
      \alpha_{i}-\partial_{i}f(\mathbf{x},\mathbf{z})=0,\quad\mathbf{x}\in U_i(\mathbf{z}),\notag\\
      \mathcal{L}^{\mathbf{z}}f(\mathbf{x},\mathbf{z})+\sum_{l=k+1}^N\lambda_{l}(\mathbf{z})f(\mathbf{x}^{(l)},\mathbf{z}^{l})=0, \quad\mathbf{x}\in U(\mathbf{z}),\end{aligned}\right.
\end{align}
where we have introduced
\begin{align}
  U_i(\mathbf{z}):=\big\{x_i\geq m_i(\mathbf{z})\big\},\quad\text{and}\quad U(\mathbf{z}):=\bigcap_{i=k+1}^NU^c_i(\mathbf{z}).
\end{align}
For $\mathbf{z}=0^{j_1,\ldots,j_k}$ and $\mathbf{x}=(x_1,\ldots,x_N)$ with $x_{j_i}=0$, $1\leq i\leq k$, we can define $U_i(\mathbf{z})$, $U(\mathbf{z})$ and the operator $\mathcal{L}^{\mathbf{z}}$ in the same manner as \eqref{U-i-and-U} and \eqref{L^z}, except that the notation $i$ and $l$ in the expression, satisfying $k+1\leq i,l\leq N$, is replaced with $j_i$ and $j_l$, satisfying $k+1\leq i,l\leq N$.

With the discussion and notations above, we now proceed to prove by induction that the following statement {\bf(S$_n$)} holds, for $1\leq n\leq N$:
\ \\
\begin{itemize}
\item[{\bf(S$_n$)}] For $N-n\leq k\leq N$ and $\mathbf{z}=0^{j_1,\ldots,j_k}$, there exists a solution $f$ to HJBVI
\begin{align}\label{original HJB1}
\max_{k+1\leq i\leq N}\left\{\mathcal{L}^{\mathbf{z}}f(\mathbf{x},\mathbf{z})+\sum_{l=k+1}^N\lambda_{j_l}(\mathbf{z})f(\mathbf{x}^{(j_l)},\mathbf{z}^{j_l}),\alpha_{j_i}-\partial_{j_i}f(\mathbf{x},\mathbf{z})\right\}=0,
\end{align}
where $f$ admits the form $f(\mathbf{x},\mathbf{z})=\sum_{i=k+1}^Nf_{j_i}(x_{j_i} ,\mathbf{z})$, satisfying
\begin{align}\label{solution-decomposition}
f_{j_i}(x,\mathbf{z})=\left\{\begin{aligned}
f_{j_i,1}(x,\mathbf{z})+C_{j_i}(\mathbf{z})f_{j_i,2}(x,\mathbf{z}),\quad &0\leq x\leq m_{j_i}(\mathbf{z}),\\
f_{j_i,1}(m_{j_i}(\mathbf{z}),\mathbf{z})+C_{j_i}(\mathbf{z})f_{j_i,2}(m_{j_i}(\mathbf{z}),\mathbf{z})+\alpha_{j_i}(x-m_{j_i}(\mathbf{z})),\quad &x\geq m_{j_i}(\mathbf{z}).
\end{aligned}\right.
\end{align}
In particular, for $k+1\leq i\leq N$,
\begin{align}\label{0-region}
\left\{\begin{aligned}
      \alpha_{j_i}-\partial_{j_i}f(\mathbf{x},\mathbf{z})=0,\quad\mathbf{x}\in U_i(\mathbf{z}),\\
      \mathcal{L}^{\mathbf{z}}f(\mathbf{x},\mathbf{z})+\sum_{l=k+1}^N\lambda_{j_l}(\mathbf{z})f(\mathbf{x}^{(j_l)},\mathbf{z}^{j_l})=0, \quad\mathbf{x}\in U(\mathbf{z}),\end{aligned}\right.
\end{align}
and $f_{j_i}(0,\mathbf{z})=0$, $f_{j_i}\geq0$, $f'_{j_i}>0$, $f''_{j_i}\leq0$, $\lim_{x\rightarrow+\infty}f_{j_i}(x,\mathbf{z})=+\infty$.
\end{itemize}

The expressions of \eqref{func1} and \eqref{fzi expression}, Proposition \ref{THHJB} and Theorem \ref{soltwo} in the previous section imply that {\bf(S$_n$)} holds when $n=1,2$.

Let $n$ be any fixed integer satisfying $1\leq n<N$. Assuming that statement {\bf(S$_n$)} holds true, we continue to show by induction that statement {\bf(S$_{n+1}$)} is also true. Due to symmetry, it suffices to show that HJBVI \eqref{original HJB1} admits a solution $f(\mathbf{x},\mathbf{z})$, for $\mathbf{z}=0^{1,\ldots, k}$ when $k=N-n-1$, as well as that $f(\mathbf{x},\mathbf{z})$ should admit the form specified in \eqref{solution-decomposition} and \eqref{0-region}. In the case where $\mathbf{z}=0^{1,\ldots, k}$ and $k=N-n-1$, the previous HJBVI \eqref{original HJB1} turns out to be
\begin{align}\label{induct HJB}
\max_{N-n\leq i\leq N}\left\{\mathcal{L}^{\mathbf{z}}f(\mathbf{x},\mathbf{z})+\sum_{l=N-n}^N\left(\sum_{j\neq l}\lambda_l(\mathbf{z})f_{j}(x_{j},\mathbf{z}^l)\right),\alpha_i-\partial_if(\mathbf{x},\mathbf{z})\right\}=0.
\end{align}

In the same fashion of the previous section with two subsidiaries, it is sufficient to study the auxiliary variational inequality, for $N-n\leq i\leq N$, with one dimensional variable $x\geq 0$ that
\begin{align}\label{original HJB3}
\max\left\{\mathcal{A}^{\mathbf{z},i}f_i(x,\mathbf{z})+\left(\sum_{\substack{l=N-n\\l\neq i}}^N\lambda_l(\mathbf{z})f_{i}(x,\mathbf{z}^l)\right),\alpha_i-f_i'(x,\mathbf{z})\right\}=0.
\end{align}
Here, we define the operator
\begin{align*}
\mathcal{A}^{\mathbf{z},i}f:=-\left(r+\tilde{\lambda}(\mathbf{z})\right)f+a_if'+\frac12b^2_if'',
\end{align*}
where $\tilde{\lambda}(\mathbf{z}):=\sum_{l=N-n}^N\lambda_l(\mathbf{z})$.

\begin{lem}\label{multi-fi}
Suppose that statement {\bf(S$_n$)} is true, then the auxiliary variational inequality \eqref{original HJB3} with the boundary condition $f(\mathbf{0},\mathbf{z})=0$ admits a $C^2$ solution $f_{i}(x,\mathbf{z})$, $N-n\leq i\leq N$, where $\mathbf{z}=0^{1,\ldots,N-n-1}$, and
\begin{align}\label{fzi expression1}
f_{i}(x,\mathbf{z})=\left\{\begin{aligned}
f_{i,1}(x,\mathbf{z})+C_i(\mathbf{z})f_{i,2}(x,\mathbf{z}),\quad &0\leq x\leq m_i(\mathbf{z}),\\
f_{i,1}(m_i(\mathbf{z}),\mathbf{z})+C_i(\mathbf{z})f_{i,2}(m_i(\mathbf{z}),\mathbf{z})+\alpha_i(x-m_i(\mathbf{z})),\quad &x> m_i(\mathbf{z}).
\end{aligned}\right.
\end{align}
Moreover, for $x\geq0$ and $N-n\leq i\leq N$, it holds that
\begin{align}\label{HJBVI-specific1}
\left\{\begin{aligned}
\mathcal{A}^{\mathbf{z},i}f_i(x,\mathbf{z})+\left(\sum_{\substack{l=N-n\\l\neq i}}^N\lambda_l(\mathbf{z})f_{i}(x,\mathbf{z}^l)\right)=0,\quad &x\in[0,m_i(\mathbf{z})],\\
\alpha_i-f_i'(x,\mathbf{z})=0,\quad &x\in[m_i(\mathbf{z}),+\infty),
\end{aligned}\right.
\end{align}
as well as that $f_{i}(0,\mathbf{z})=0$, $f'_{i}(x,\mathbf{z})>0$, $f''_{i}(x,\mathbf{z})\leq0$, and $\lim_{x\rightarrow+\infty}f_i(\mathbf{x},\mathbf{z})=+\infty$.
\end{lem}

\begin{proof}
Note that for any $N-n\leq l\leq N$, $\mathbf{z}^l=0^{1,\ldots,N-n-1,l}$. Our induction assumption {\bf(S$_n$)} gives the boundary condition $\sum_{l\neq i}\lambda_l(\mathbf{z})f_{i}(0,\mathbf{z}^l)=0$ as well as the results
\begin{align*}
\sum_{l\neq i}\lambda_l(\mathbf{z})f_{i}(x,\mathbf{z}^l)\geq0,\ \ \ \left(\sum_{l\neq i}\lambda_l(\mathbf{z})f_{i}(x,\mathbf{z}^l)\right)'>0,\ \ \ \left(\sum_{l\neq i}\lambda_l(\mathbf{z})f_{i}(x,\mathbf{z}^l)\right)''\leq0,
\end{align*}
for $ N-n\leq i\leq N$. Therefore, for $N-n\leq i\leq N$, we can conclude the existence of $C^2$ solution $f_{i}(x,\mathbf{z})$ by using the same argument in the proof of Proposition \ref{THHJB} and obtain the existence of free boundary points $m_i(\mathbf{z})$ with $\mathbf{z}=0^{1,\ldots, N-n-1}$ such that \eqref{HJBVI-specific1} holds. Moreover, we have $f_{i}(0,\mathbf{z})=0$, $f'_{i}(x,\mathbf{z})>0$, $f''_{i}(x,\mathbf{z})\leq0$, $x\geq0$. In view of \eqref{fzi expression1}, we also have $\lim_{x\rightarrow+\infty}f_i(x,\mathbf{z})=+\infty$.
\end{proof}

\begin{lem}\label{n comp HJB}
Suppose that statement {\bf(S$_n$)} is true, then the variational inequality \eqref{induct HJB} admits a $C^2$ solution, which is in the separation form of
\begin{align}\label{4.1solu}
  f(\mathbf{x},\mathbf{z})=\sum_{i=N-n}^Nf_{i}(x_i,\mathbf{z}),
\end{align}
where each $f_{i}(x,\mathbf{z})$ defined in \eqref{fzi expression1} is the solution to the auxiliary variational inequality \eqref{original HJB3}. In particular, for $x\geq0$, $f_{i}(x,\mathbf{z})$ satisfies \eqref{HJBVI-specific1}, $f_{i}(0,\mathbf{z})=0$, $f'_{i}(x,\mathbf{z})>0$, $f''_{i}(x,\mathbf{z})\leq0$, and $\lim_{x\rightarrow+\infty}f_i(\mathbf{x},\mathbf{z})=+\infty$. Therefore statement {\bf(S$_{n+1}$)} is also true.
\end{lem}

\begin{proof}
It suffices to investigate the $C^2$ solution of the variational inequality \eqref{induct HJB}. Let $f$ be the function defined in \eqref{4.1solu}. It is then obvious that $f$ is $C^2$. In view of \eqref{original HJB3}, we have
\begin{align*}
\mathcal{L}^{\mathbf{z}}f(\mathbf{x},\mathbf{z})+\sum_{i=N-n}^N\left(\sum_{l\neq i}\lambda_l(\mathbf{z})f_{i}(x_{i},\mathbf{z}^l)\right)&=\sum_{i=N-n}^N\mathcal{A}^{\mathbf{z},i}f_{i}(x_i,\mathbf{z})+\sum_{i=N-n}^N\left(\sum_{l\neq i}\lambda_l(\mathbf{z})f_{i}(x_{i},\mathbf{z}^l)\right)\notag\\
&=\sum_{i=N-n}^N\left(\mathcal{A}^{\mathbf{z},i}f_{i}(x_i,\mathbf{z})+\left(\sum_{l\neq i}\lambda_l(\mathbf{z})f_{i}(x_{i},\mathbf{z}^l)\right)\right)\leq0.
\end{align*}
Furthermore, $\alpha_i-\partial_if(\mathbf{x},\mathbf{z})=\alpha_i-f'_{i}(x_i,\mathbf{z})\leq0$, $i=N-n,\ldots,N$. It follows that
\begin{align}\label{maxleq0}
\max_{N-n\leq i\leq N}\left\{\mathcal{L}^{\mathbf{z}}f(\mathbf{x},\mathbf{z})+\sum_{i=N-n}^N\left(\sum_{l\neq i}\lambda_l(\mathbf{z})f_{i}(x_{i},\mathbf{z}^l)\right),\alpha_i-\partial_if(\mathbf{x},\mathbf{z})\right\}\leq0.
\end{align}
Now we claim that
\begin{align*}
\max_{N-n\leq i\leq N}\left\{\mathcal{L}^{\mathbf{z}}f(\mathbf{x},\mathbf{z})+\sum_{i=N-n}^N\left(\sum_{l\neq i}\lambda_l(\mathbf{z})f_{i}(x_{i},\mathbf{z}^l)\right),\alpha_i-\partial_if(\mathbf{x},\mathbf{z})\right\}=0.
\end{align*}
Fix $x_i\geq 0$, $N-n\leq i\leq N$ and $\mathbf{z}=0^{1,\ldots, N-n-1}$. If $\mathcal{L}^{\mathbf{z}}f(\mathbf{x},\mathbf{z})+\sum_{i=N-n}^N\left(\sum_{l\neq i}\lambda_l(\mathbf{z})f_{i}(x_{i},\mathbf{z}^l)\right)=0$, then the equality trivially holds. If $\mathcal{L}^{\mathbf{z}}f(\mathbf{x},\mathbf{z})+\sum_{i=N-n}^N\left(\sum_{l\neq i}\lambda_l(\mathbf{z})f_{i}(x_{i},\mathbf{z}^l)\right)<0$ , it follows that $\mathcal{A}^{\mathbf{z},i}f_{i}(x_i,\mathbf{z})+\left(\sum_{l\neq i}\lambda_l(\mathbf{z})f_{i}(x_{i},\mathbf{z}^l)\right)<0$, for some $i$. As $f_{i}$ is chosen to solve \eqref{original HJB3}, it holds that $\alpha_i-\partial_if(\mathbf{x},\mathbf{z})=\alpha_i-f'_{i}(x_i,\mathbf{z})=0$. Therefore, our claim holds that $f(\mathbf{x},\mathbf{z})$ is the $C^2$ solution to the variational inequality \eqref{induct HJB}. Moreover, for $x\geq0$, we have by Lemma \ref{multi-fi} that $f_{i}(x,\mathbf{z})$ defined in \eqref{fzi expression1} satisfies $f_{i}(0,\mathbf{z})=0$, $f'_{i}(x,\mathbf{z})>0$, $f''_{i}(x,\mathbf{z})\leq0$ and $\lim_{x\rightarrow+\infty}f_i(x,\mathbf{z})=+\infty$. Meanwhile, \eqref{HJBVI-specific1} in Lemma \ref{multi-fi} yields the desired property in \eqref{0-region}.

Given the results above, we conclude that, for $\mathbf{z}=0^{1,\ldots,N-n-1}$, HJBVI \eqref{original HJB1} has a solution $f(\mathbf{x},\mathbf{z})$, which admits the form in \eqref{solution-decomposition} and \eqref{0-region}. This completes the proof of the statement {\bf(S$_{n+1}$)}.
\end{proof}

By mathematical induction, we can present the following main result.
\begin{thm}\label{induction-result}
Statement {\bf(S$_N$)} is true. In particular, for $0\leq k\leq N$ and $\mathbf{z}=0^{1,\ldots,k}$, the recursive system of HJBVI \eqref{original HJB1} admits a $C^2$ solution in the separation form of
\begin{align}\label{multi-solu-f}
  f(\mathbf{x},\mathbf{z})=\sum_{i=k+1}^Nf_{i}(x_i,\mathbf{z}),
\end{align}
where each $f_{i}(x,\mathbf{z})$ is defined in \eqref{fzi expression1}, with $n=N-1$, i.e., $f_{i}(x,\mathbf{z})$ is the solution to the auxiliary variational inequality \eqref{original HJB3} and satisfies \eqref{HJBVI-specific1}, $k+1\leq i\leq N$.
\end{thm}

\begin{rem}\label{value_function&Sigma}
It can be observed from \eqref{original HJB3} that each function $f_i(x_i,\mathbf{z})$ in the separation form \eqref{4.1solu} is actually independent of the correlation coefficient matrix $\Sigma$. Therefore, the solution $f(\mathbf{x},\mathbf{z})$ to the recursive system of HJBVI \eqref{original HJB1}, for $0\leq k\leq N$, is also independent of the correlation coefficient matrix $\Sigma=(\rho_{ij})_{N\times N}$.
\end{rem}

\section{Proof of Verification Theorem}\label{sec:verf}

In this section, we construct the optimal dividend strategy using the $C^2$ solution of the recursive system HJBVI \eqref{original HJB1} and complete the proof of the main theorem.

\begin{proof}[Proof of Theorem \ref{mainthm1}]\ \\
Thanks to Theorem \ref{induction-result}, we can readily conclude that variational inequality \eqref{varineqNc} for the case $k=0$ (i.e. $\mathbf{z}=\mathbf{0}$ and $N$ subsidiaries are alive) also admits the $C^2$ solution in the separation form \eqref{multi-solu-f}. Moreover, as statement {\bf(S$_N$)} holds, the existence of mapping $m_{j_i}(\mathbf{z}):\{0,1\}^N\mapsto(0,+\infty)$ is also guaranteed, for any $\mathbf{z}=0^{j_1,\ldots,j_k}$, $1\leq i\leq k$ as well as $\mathbf{z}=\mathbf{0}$.

Let $\tau$ be an arbitrary stopping time, and $\mathbf{D}(t)=(D_1(t),\ldots,D_N(t))$ be an arbitrary admissible strategy. By using It\^{o}'s formula, we first get
\begin{align}\label{itoN}
&\sum_{i=1}^N\alpha_i\int_0^\tau e^{-rs}dD_i(s)+e^{-r\tau}f\left(\mathbf{X}(\tau),\mathbf{Z}(\tau)\right)-f(\mathbf{x},\mathbf{z})\notag\\
=&\int_0^{\tau}e^{-rs}\left[\mathcal{L}^{\mathbf{Z}(s)}f(\mathbf{X}(s),\mathbf{Z}(s))+\sum_{l=k+1}^N\lambda_l(\mathbf{Z}(s))f(\mathbf{X}^{(l)}(s),\mathbf{Z}^{l}(s))\right]ds\notag\\
&+\sum_{i=1}^N\int_0^{\tau}e^{-rs}\left[\alpha_i-\partial_if(\mathbf{X}(s), \mathbf{Z}(s))\right]dD^c_i(s)\notag\\
&+\sum_{0<s\leq\tau,\Delta Z(s)\neq0}e^{-rs}\sum_{j=1}^N\Delta Z_j(s)\bigg[f\left(\mathbf{X}^{(j)}(s-)-\Delta \mathbf{D}^j(s),\mathbf{Z}^j(s-)\right)-f\left(\mathbf{X}^{(j)}(s-),\mathbf{Z}^j(s-)\right)\notag\\
&+\sum_{\substack{i=1\\i\neq j}}^N\alpha_i\Delta D_i(s)\bigg]+\sum_{0<s\leq\tau,\Delta Z(s)=0}e^{-rs}\Bigg[f\left(\mathbf{X}(s)-\Delta \mathbf{D}(s),\mathbf{Z}(s-)\right)-f\left(\mathbf{X}(s-),\mathbf{Z}(s-)\right)\notag\\
&+\sum_{i=1}^N\alpha_i\Delta D_i(s)\Bigg]+\mathcal{M}_\tau\notag\\
=:&I+II+III+IV+\mathcal{M}_\tau.
\end{align}
As $f$ solves \eqref{original HJB1}, we have that $I,II,IV\leq0$. Moreover, by noting that $f(\mathbf{x},\mathbf{z}^j)$ also solves \eqref{original HJB1}, we deduce that $III\leq0$. Note that $\mathcal{M}_{t\wedge\tau}$ is a local martingale. There exists a sequence of stopping times $\{T_n\}_{n=1}^{\infty}$ satisfying $T_n\uparrow \infty$, and
\begin{align}\label{verification-leq}
&\mathbb{E}\left[\sum_{i=1}^{N} \alpha_i \int_0^{\tau} e^{-rs}dD_i(s)\right]\notag\\
\leq& \lim_{n\rightarrow\infty}\mathbb{E}\left[ \sum_{i=1}^{N} \alpha_i \int_0^{\tau\wedge T_n} e^{-rs}dD_i(s)+e^{-r(\tau\wedge T_n)}f(\mathbf{X}(\tau\wedge T_n),\mathbf{Z}(\tau\wedge T_n)) \right]\notag\\
\leq & f(\mathbf{x},\mathbf{z})+\lim_{n\rightarrow\infty}\mathbb{E}[\mathcal{M}_{\tau\wedge T_n}]=f(\mathbf{x},\mathbf{z}).
\end{align}

In view that $\mathbf{D}(t)$ is arbitrary, we obtain by sending $\tau$ in \eqref{verification-leq} to $+\infty$ that
\begin{align}\label{verification-leq1}
  \sup_{D} J(\mathbf{x},\mathbf{z},\mathbf{D})\leq f(\mathbf{x},\mathbf{z}).
\end{align}
Let us continue to prove that ``$=$'' holds in \eqref{verification-leq1}. Consider the c\`{a}dl\`{a}g strategy
\begin{align*}
D^*_i(t)&:=\max\left\{0,\sup_{0\leq s\leq t}\left\{\wdt X_i(s)-m_i\left(\mathbf{Z}(s)\right)\right\}\right\},\notag\\
X^*_i(t)&=\wdt X_i(t)-D^*_i(t).
\end{align*}
We set $A_i(t):=\mathbf{1}_{\left\{D^*_i(t)=\wdt X_i(t)-m_i\left(\mathbf{Z}(t)\right)\right\}}$. It follows that
\begin{align}
\label{region1}X^*_i(t)&=\wdt X_i(t)-D^*_i(t)\leq m_i\left(\mathbf{Z}(t)\right),\\
dD^*_i(t)&=A_i(t)dD^*_i(t).
\end{align}
On $\left\{D^*_i(t)=\wdt X_i(t)-m_i\left(\mathbf{Z}(t)\right)\right\}$, we have that
\begin{align*}
X^*_i(t)=\tilde{X}_i(t)-D^*_i(t)=m_i\left(\mathbf{Z}(t)\right),
\end{align*}
and vise versa. It then follows that
\begin{align*}
dD^*_i(t)=A_i(t)dD^*_i(t)=\mathbf{1}_{\left\{X^*_i(t)=m_i\left(\mathbf{Z}(t)\right)\right\}}dD^*_i(t).
\end{align*}
Furthermore, we have on $\left\{X^*_i(t)=m_i\left(\mathbf{Z}(t)\right)\right\}$ that
\begin{align}\label{jumpsize}
X^*_i(t-)=X^*_i(t)+\Delta D^*_i(t)\geq X^*_i(t)=m_i\left(\mathbf{Z}(t)\right).
\end{align}
In view of \eqref{region1}, \eqref{0-region}, we have that
\begin{align}\label{optimal1}
\mathcal{L}^{\mathbf{Z}(s)}f(\mathbf{X}^*(s),\mathbf{Z}(s))+\sum_{l=k+1}^N\lambda_l(\mathbf{Z}(s))f((\mathbf{X}^*)^{(l)}(s),\mathbf{Z}^{l}(s))=0.
\end{align}
Note that for $x_i\geq m_i\left(\mathbf{z}\right)$, $\partial_if(\mathbf{x},\mathbf{z})=f'_i(x_i, \mathbf{z})=\alpha_i$. Hence, it holds that $\partial_if(\mathbf{X}^*(s),\mathbf{Z}(s))=\alpha_i$ on $\left\{X^*_i(t)=m_i\left(\mathbf{Z}(t)\right)\right\}$, which then entails that
\begin{align}\label{optimal2}
&\sum_{i=1}^N\int_0^{\tau}e^{-rs}\left[\alpha_i-\partial_if(\mathbf{X}^*(s), \mathbf{Z}(s))\right](D^*_i)^c(s)\\
=&\sum_{i=1}^N\int_0^{\tau}e^{-rs}\left[\alpha_i-\partial_if(\mathbf{X}^*(s), \mathbf{Z}(s))\right]\mathbf{1}_{\left\{X^*_i(t)=m_i\left(\mathbf{Z}(t)\right)\right\}}d(D^*_i)^c(s)=0.
\end{align}
By virtue of \eqref{jumpsize}, we can see that whenever $\Delta D^*_i(s)\neq0$, it holds that $X^*_i(s-)>X^*_i(s-)-\Delta D^*_i(s)=X^*_i(s)=m_i\left(\mathbf{Z}(s)\right)$. By using the fact that $\partial_if(\mathbf{x},\mathbf{z})=f'_i(x_i,\mathbf{z})=\alpha_i$, for $x_i\geq m_i\left(\mathbf{z}\right)$, again, we obtain that
{\small
\begin{align}\label{optimal3}
&\sum_{j=1}^N\Delta Z_j(s)\bigg[f\left((\mathbf{X}^*)^{(j)}(s-)-\Delta (\mathbf{D}^*)^{(j)}(s),\mathbf{Z}^j(s-)\right)-f\left((\mathbf{X}^*)^{(j)}(s-),\mathbf{Z}^j(s-)\right)+\sum_{\substack{i=1\\i\neq j}}^N\alpha_i\Delta D^*_i(s)\bigg]\notag\\
=&\sum_{j=1}^N\Delta Z_j(s)\bigg[f\left((\mathbf{X}^*)^{(j)}(s-)-\Delta (\mathbf{D}^*)^{(j)}(s),\mathbf{Z}(s)\right)-f\left((\mathbf{X}^*)^{(j)}(s-),\mathbf{Z}(s)\right)+\sum_{\substack{i=1\\i\neq j}}^N\alpha_i\Delta D^*_i(s)\bigg]\notag\\
=&0.
\end{align}}Similarly, we obtain the equality that
\begin{align}\label{optimal4}
&\quad\sum_{0<s\leq\tau,\Delta Z(s)=0}e^{-rs}\left[f\left(\mathbf{X}^*(s-)-\Delta \mathbf{D}^*(s),\mathbf{Z}(s-)\right)-f\left(\mathbf{X}^*(s-),\mathbf{Z}(s-)\right)+\sum_{i=1}^N\alpha_i\Delta D^*_i(s)\right]\notag\\
&=\sum_{0<s\leq\tau,\Delta Z(s)=0}e^{-rs}\left[f\left(\mathbf{X}^*(s-)-\Delta \mathbf{D}^*(s),\mathbf{Z}(s)\right)-f\left(\mathbf{X}^*(s-),\mathbf{Z}(s)\right)+\sum_{i=1}^N\alpha_i\Delta D^*_i(s)\right]\notag\\
&=0.\notag\\
\end{align}
Putting all the pieces together, we conclude from \eqref{itoN} and \eqref{optimal1}-\eqref{optimal4} that
\begin{align}\label{optimal8}
\sum_{i=1}^N\alpha_i\int_0^\tau e^{-rs}dD^*_i(s)+e^{-r\tau}f\left(\mathbf{X}^*(\tau),\mathbf{Z}(\tau)\right)-f(\mathbf{x},\mathbf{z})=\mathcal{M}_{\tau},\quad\tau\geq0,
\end{align}
where $\mathcal{M}_{\tau}$ is a local martingale. Hence, there exists a sequence of stopping times $\{T_n\}_{n=1}^\infty$ satisfying $T_n\uparrow\infty$, and
\begin{align}\label{optimal5}
&\mathbb{E}\left[\sum_{i=1}^N\alpha_i\int_0^{\tau\wedge T_n} e^{-rs}dD^*_i(s)+e^{-r({\tau\wedge T_n})}f\left(\mathbf{X}^*(\tau\wedge T_n),\mathbf{Z}(\tau\wedge T_n)\right)\right]-f(\mathbf{x},\mathbf{z})\notag\\
=&\mathbb{E}\left[\mathcal{M}_{\tau\wedge T_n}\right]=0.
\end{align}
In view of \eqref{region1}, we have $0\leq X^*_i(\tau)\leq m_i\left(\mathbf{Z}(\tau)\right)$, $\tau\geq0$, which entails that $X^*_i(\tau)$ is a bounded process. It follows that $f\left(\mathbf{X}^*(\tau),\mathbf{Z}(\tau)\right)$ is also bounded. Note that
\begin{align*}
\lim_{n\rightarrow\infty}e^{-r(\tau\wedge T_n)}f\left(\mathbf{X}^*(\tau\wedge T_n),\mathbf{Z}(\tau\wedge T_n)\right)=e^{-r\tau}f\left(\mathbf{X}^*(\tau),\mathbf{Z}(\tau)\right)\quad\text{a.s..}
\end{align*}
By passing the limit in \eqref{optimal5}, we arrive at
\begin{align}\label{optimal6}
\mathbb{E}\left[\sum_{i=1}^N\alpha_i\int_0^{\tau} e^{-rs}dD^*_i(s)+e^{-r\tau}f\left(\mathbf{X}^*(\tau),\mathbf{Z}(\tau)\right)\right]-f(\mathbf{x},\mathbf{z})=0.
\end{align}
Note that $\lim_{\tau\rightarrow+\infty}e^{-r\tau}f\left(\mathbf{X}^*(\tau),\mathbf{Z}(\tau)\right)=0$ a.s.. Sending $\tau$ to  $+\infty$ in \eqref{optimal6} yields that
\begin{align}\label{optimal7}
\mathbb{E}\left[\sum_{i=1}^N\alpha_i\int_0^{\tau_i} e^{-rs}dD^*_i(s)\right]-f(\mathbf{x},\mathbf{z})=0,
\end{align}
which completes the proof.
\end{proof}
\begin{rem}\label{rem5-1}
Similar to the derivation of \eqref{optimal8}, for $i=1,\ldots,N$, if we extend the definition of $f_i$ in such a way that $f_i(x_i,\mathbf{z})=0$ whenever the $i$-th component of $\mathbf{z}$ is $1$, then, following the proof of Theorem \ref{mainthm1} and using \eqref{original HJB3}, we can show
\begin{align*}
\alpha_i\int_0^\tau e^{-rs}dD^*_i(s)+e^{-r\tau}f_i\left(X^*_i(\tau),\mathbf{Z}(\tau)\right)-f_i(x_i,\mathbf{z})=\widetilde{\mathcal M}^{(i)}_{\tau},\quad i=1,\ldots,N,
\end{align*}
where $\widetilde{\mathcal M}^{(i)}_{\tau}$ are local martingales, for $x_i\in[0,+\infty)$, $i=1,\ldots,N,$ and $\mathbf{z}=\mathbf{0}$. In the same fashion to obtain \eqref{optimal7}, one can also get
\begin{align*}
\mathbb{E}\left[\alpha_i\int_0^{\tau_i} e^{-rs}dD^*_i(s)\right]-f_i(x_i,\mathbf{z})=0,\quad i=1,\ldots,N.
\end{align*}
This equality implies a natural linear separation form of $f(\mathbf{x},\mathbf{z})$ in \eqref{4.1solu} because we can see that
\begin{align*}
f(\mathbf{x},\mathbf{z})={\mathbb E}\left(\sum^N_{i=1}\alpha_i\int_0^{\tau_i}e^{-rt} dD^*_i(t)\right)=\sum^N_{i=1}{\mathbb E}\left[\alpha_i\int_0^{\tau_i}e^{-rt} dD^*_i(t)\right],
\end{align*}
and each $f_i(x_i,\mathbf{z})$ stands for the expected value that $f_i(x_i,\mathbf{z})={\mathbb E}\left[\alpha_i\int_0^{\tau_i}e^{-rt} dD^*_i(t)\right]$ given the optimal dividend policy $D_i^*$ for the subsidiary i.  However, we also point out that $D_i^*$ is the i-th component of the optimal control $\mathbf{D}^*$ which solves the group dividend problem. One can not simply interpret that $f_i(x_i,\mathbf{z})$ is the value function or $D_i^*$ is the optimal control when we purely solve a dividend optimization problem for the single subsidiary i without taking account all other subsidiaries. The vector process $\mathbf{D}^*$ is the solution that is optimal for a whole group and it has a coupled nature because the variational inequality \eqref{original HJB3} or the solution form \eqref{fzi expression1} for each $f_i(x_i,\mathbf{z})$ depends on the default intensities of all surviving subsidiaries and the value functions given that one more subsidiary has defaulted.
\end{rem}

\section{Conclusions}
We formulate and investigate an optimal dividend problem for a multi-line insurance group. Each subsidiary within the group runs a product line and all subsidiaries are exposed to some external contagious default risk. By using the backward recursive scheme and the smooth-fit principle, the associated recursive system of HJBVIs is studied and the value function of the expected total dividend is proved to be its classical solution that has a separation form. We verify that the optimal dividend fits the type of barrier control and the barrier for each surviving subsidiary is dynamically modulated by the default state.

Some future research can be conducted along different directions. Firstly, one can consider the more general model of $\hat{X}_i$ with jumps such as the classical Cram\'er-Lundberg model or other jump-diffusion models. Secondly, we note that the real life default events from credit assets can hurt the surplus management but may not lead to domino bankruptcies of subsidiaries due to strict regulations of the whole insurance sector. It is more realistic to consider the problem when $Z_i(t)$ can take values in $[0,1]$ so that the default event only leads to a large size downward jump of the surplus process and certain recovery rate can be incorporated. Moreover, the default intensity $\lambda_i\left(\mathbf{Z}(t), X_i(t)\right)$ of $Z_i(t)$ may also depend on the surplus level $X_i(t)$ of the $i$-th subsidiary to depict the situation that a larger surplus level guarantees a smaller default probability. The inclusion of these factors will complicate the analysis of HJBVIs significantly because the backward induction can not be applied in a simple way and it is an open problem whether the optimal dividend of each subsidiary is still of the barrier type. It will be interesting to study these model extensions by applying some distinctive PDE arguments. Another appealing future work is to accommodate the collaborating bail-out dividend (see \cite{AlbrecherAM}, \cite{GuSZ} and \cite{Grandits}) in the present setting with contagious default risk so that each subsidiary can perform capital injection to other subsidiaries within the group whenever their financial ruins or credit default events happen.

\appendix
\section{Appendix: Derivation of \eqref{original HJB}}\label{appx}

For the default process starting from $\mathbf{Z}(0)=\mathbf{z}=(0,0)$, we present here the argument to derive the associated HJBVI using It\^{o}'s lemma. For a given function $\psi(\cdot,\mathbf{z})\in C^2(\mathbb{R}^2)$, let us rewrite
\begin{align}\label{ito}
&\alpha_1\int_0^\tau e^{-rs}dD_1(s)+\alpha_2\int_0^\tau e^{-rs}dD_2(s)+e^{-r\tau}\psi\left(\mathbf{X}(\tau),\mathbf{Z}(\tau)\right)-\psi(\mathbf{x},\mathbf{z})\notag\\
&=\int_0^{\tau}e^{-rs}\tilde{\mathcal{L}}^{(0,0)}\psi(s)ds+\int_0^{\tau}e^{-rs}\left[\alpha_1-\partial_1\psi(s)\right]dD^c_1(s)+\int_0^{\tau}e^{-rs}\left[\alpha_2-\partial_2\psi(s)\right]dD^c_2(s)\notag\\
&+\alpha_1\int_0^\tau e^{-rs}dD_1(s)+\alpha_2\int_0^\tau e^{-rs}dD_2(s)\notag\\
&+\sum_{0<s\leq\tau}e^{-rs}\left[\psi\left(\mathbf{X}(s),\mathbf{Z}(s)\right)-\psi\left(\mathbf{X}(s-),\mathbf{Z}(s-)\right)\right]+\mathcal{M}_\tau\notag\\
&=\int_0^{\tau}e^{-rs}\tilde{\mathcal{L}}^{(0,0)}\psi(s)ds+\int_0^{\tau}e^{-rs}\left[\alpha_1-\partial_1\psi(s)\right]dD^c_1(s)+\int_0^{\tau}e^{-rs}\left[\alpha_2-\partial_2\psi(s)\right]dD^c_2(s)\notag\\
&+\alpha_1\int_0^\tau e^{-rs}dD_1(s)+\alpha_2\int_0^\tau e^{-rs}dD_2(s)\notag\\
&+\sum_{0<s\leq\tau,\Delta Z(s)\neq0}e^{-rs}\left[\psi\left(\mathbf{X}(s),\mathbf{Z}(s)\right)-\psi\left(\mathbf{X}(s-),\mathbf{Z}(s-)\right)\right]\notag\\
&+\sum_{0<s\leq\tau,\Delta Z(s)=0}e^{-rs}\left[\psi\left(\mathbf{X}(s)+\Delta D(s),\mathbf{Z}(s-)\right)-\psi\left(\mathbf{X}(s-),\mathbf{Z}(s-)\right)\right]+\mathcal{M}_\tau\notag\\
&=\int_0^{\tau}e^{-rs}\mathcal{L}^{(0,0)}\psi(s)ds+\int_0^{\tau}e^{-rs}\left[\alpha_1-\partial_1\psi(s)\right]dD^c_1(s)+\int_0^{\tau}e^{-rs}\left[\alpha_2-\partial_2\psi(s)\right]dD^c_2(s)\notag\\
&+\sum_{0<s\leq\tau,\Delta Z(s)\neq0}e^{-rs}\Delta Z_1(s)\left[\psi\left(0,X_2(s-)-\Delta D_2(s),(1,0)\right)-\psi\left(\mathbf{X}(s-),(0,0)\right)+\alpha_2\Delta D_2(s)\right]\notag\\
&+\sum_{0<s\leq\tau,\Delta Z(s)\neq0}e^{-rs}\Delta Z_2(s)\left[\psi\left(X_1(s-)-\Delta D_1(s),0,(0,1)\right)-\psi\left(\mathbf{X}(s-),(0,0)\right)+\alpha_1\Delta D_1(s)\right]\notag\\
&+\sum_{0<s\leq\tau,\Delta Z(s)=0}e^{-rs}\Big[\psi\left(\mathbf{X}(s)-\Delta D(s),\mathbf{Z}(s-)\right)-\psi\left(\mathbf{X}(s-),\mathbf{Z}(s-)\right)\notag\\
&+\alpha_1\Delta D_1(s)+\alpha_2\Delta D_2(s)\Big]+\mathcal{M}_\tau,
\end{align}
where $\mathcal{M}_\tau$ is a local martingale.

Let us turn to the jump terms. According to assumptions that no simultaneous jumps can occur in the sense of \eqref{simul_jump_assumption1} and \eqref{simul_jump_assumption2}, it follows that
\begin{align*}
\Delta Z_1(s)\Delta D_1(s)=\Delta Z_2(s)\Delta D_2(s)=\Delta Z_1(s)\Delta Z_2(s)=0.
\end{align*}

On $\left\{\Delta Z(s)\neq0\right\}$, let us consider $\mathbf{Z}(s-)=(0,0)$. We have
\begin{align*}
&e^{-rs}\left[\psi\left(\mathbf{X}(s),\mathbf{Z}(s)\right)-\psi\left(\mathbf{X}(s-),\mathbf{Z}(s-)\right)\right]\notag\\
=&e^{-rs}\Delta Z_1(s)\left[\psi\left((0,X_2(s-)-\Delta D_2(s)),(1,0)\right)-\psi\left(\mathbf{X}(s-),(0,0)\right)\right]\notag\\
&+e^{-rs}\Delta Z_2(s)\left[\psi\left((X_1(s-)-\Delta D_1(s),0),(0,1)\right)-\psi\left(\mathbf{X}(s-),(0,0)\right)\right],
\end{align*}
as well as
\begin{align*}
&e^{-rs}\Delta Z_1(s)\left[\psi\left((0,X_2(s-)-\Delta D_2(s)),(1,0)\right)-\psi\left(\mathbf{X}(s-),(0,0)\right)\right]\notag\\
=&e^{-rs}\Delta Z_1(s)\left[\psi\left((0,X_2(s-)-\Delta D_2(s)),(1,0)\right)-\psi\left(0,X_2(s-),(1,0)\right)\right]\notag\\
&+e^{-rs}\Delta Z_1(s)\left[\psi\left((0,X_2(s-)),(1,0)\right)-\psi\left(\mathbf{X}(s-),(0,0)\right)\right].
\end{align*}
Similarly, one can get
\begin{align*}
&e^{-rs}\Delta Z_2(s)\left[\psi\left((X_1(s-)-\Delta D_1(s),0),(0,1)\right)-\psi\left(\mathbf{X}(s-),(0,0)\right)\right]\notag\\
=&e^{-rs}\Delta Z_2(s)\left[\psi\left((X_1(s-)-\Delta D_1(s),0),(0,1)\right)-\psi\left(X_1(s-),0,(0,1)\right)\right]\notag\\
&+e^{-rs}\Delta Z_2(s)\left[\psi\left((X_1(s-),0),(0,1)\right)-\psi\left(\mathbf{X}(s-),(0,0)\right)\right].
\end{align*}
On $\left\{\Delta Z(s)=0\right\}$, we have
\begin{align*}
&e^{-rs}\left[\psi\left(\mathbf{X}(s),\mathbf{Z}(s)\right)-\psi\left(\mathbf{X}(s-),\mathbf{Z}(s-)\right)\right]\notag\\
=&e^{-rs}\left[\psi\left(\mathbf{X}(s-)-\Delta D(s),\mathbf{Z}(s-)\right)-\psi\left(\mathbf{X}(s-),\mathbf{Z}(s-)\right)\right],
\end{align*}
and also
\begin{align*}
\alpha_i\int_0^\tau e^{-rs}dD_i(s)=\sum_{0<s\leq\tau,\Delta Z_2(s)\neq0}\alpha_i e^{-rs}\Delta D_i(s)+\sum_{0<s\leq\tau,\Delta Z_2(s)=0}\alpha_i e^{-rs}\Delta D_i(s).
\end{align*}

Thanks to the martingale property in \eqref{Z_imart} and the fact that, for any $h\in C^1({\mathbb R})$ and $y\in{\mathbb R}$,
\begin{align*}
  h(y-\Delta D_i(s))-h(y)=-\int_0^{\Delta D_i(s)}h'(y-u)du,\notag
\end{align*}
we obtain the desired HJBVI \eqref{original HJB}.\\
\ \\
\textbf{Acknowledgement} X. Yu is supported by Hong Kong Early Career Scheme under no. 25302116 and The Hong Kong Polytechnic University internal grant under no. P0031417.

\ \\
\bibliographystyle{siam}
\bibliography{JLYY-1}

\begin{thebibliography}{10}

\bibitem{AlbrecherAM}
{\sc H.~Albrecher, P.~Azcue, and N.~Muler}, {\em Optimal dividend strategies
  for two collaborating insurance companies}, Advances in Applied Probability,
  49(2) (2017), pp.~515--548.

\bibitem{AlbrecherT09}
{\sc H.~Albrecher and S.~Thonhauser}, {\em Optimality results for dividend
  problems in insurance}, RACSAM-Revista de la Real Academia de Ciencias
  Exactas, Fisicas Naturales. Serie A. Matematicas, 103(2) (2009),
  pp.~295--320.

\bibitem{AminiM}
{\sc H.~Amini and A.~Minca}, {\em Inhomogeneous financial networks and
  contagious links}, Operations Research, 64(5) (2016), pp.~1109--1120.

\bibitem{Asmu}
{\sc S.~Asmussen, B.~Hjgaard, and M.~Taksar}, {\em Optimal risk control and
  dividend distribution policies. example of excess-of-loss reinsurance for an
  insurance corporation}, Finance and Stochastics, 4(3) (2000), pp.~299--324.

\bibitem{AsmussenT}
{\sc S.~Asmussen and M.~Taksar}, {\em Controlled diffusion models for optimal
  dividend pay-out}, Insurance: Mathematics and Economics, 20 (1997),
  pp.~1--15.

\bibitem{Avanzi09}
{\sc B.~Avanzi}, {\em Strategies for dividend distribution: a review}, North
  American Actuarial Journal, 13(2) (2009), pp.~217--251.

\bibitem{AvramPP07}
{\sc F.~Avram, Z.~Palmowski, and M.~R. Pistorius}, {\em On the optimal dividend
  problem for a spectrally negative l\'{e}vy process}, The Annals of Applied
  Probability, 17(1) (2007), pp.~156--180.

\bibitem{AzcueM05}
{\sc P.~Azcue and N.~Muler}, {\em Optimal investment policy and dividend
  payment strategy in an insurance company}, The Annals of Applied Probability,
  20(4) (2010), pp.~1253--1302.

\bibitem{B2}
{\sc J.~R. Birge, L.~Bo, and A.~Capponi}, {\em Risk-sensitive asset management
  and cascading defaults}, Mathematics of Operations Research, 43 (2018),
  pp.~1--28.

\bibitem{B3}
{\sc L.~Bo and A.~Capponi}, {\em Optimal investment in credit derivatives
  portfolio under contagion risk}, Mathematical Finance, 26(4) (2014),
  pp.~785--834.

\bibitem{B1}
{\sc L.~Bo, A.~Capponi, and P.~C. Chen}, {\em Credit portfolio selection with
  decaying contagion intensities}, Mathematical Finance, 29(1) (2019),
  pp.~137--173.

\bibitem{BLY2}
{\sc L.~Bo, H.~Liao, and X.~Yu}, {\em Risk-sensitive credit portfolio
  optimization under partial information and contagion risk}, Preprint,
  arXiv:1905.08004,  (2019).

\bibitem{BLY1}
\leavevmode\vrule height 2pt depth -1.6pt width 23pt, {\em Risk sensitive
  portfolio optimization with default contagion and regime-switching}, SIAM
  Journal on Control and Optimization, 57(1) (2019), pp.~366--401.

\bibitem{ChevalierVS13}
{\sc E.~Chevalier, V.~L. Vath, and S.~Scotti}, {\em An optimal dividend and
  investment control problem under debt constraints}, SIAM Journal on Financial
  Mathematics, 4(1) (2013), pp.~297--326.

\bibitem{Choulli}
{\sc T.~Choulli, M.~Taksar, and X.~Zhou}, {\em Excess-of-loss reinsurance for a
  company with debt liability and constraints on risk reduction}, Quantitative
  Finance, 1(6) (2001), pp.~573--596.

\bibitem{viscosity-guide}
{\sc M.~G. Crandall, H.~Ishii, and P.~L. Lions}, {\em User's guide to viscosity
  solutions of 2nd order partial differential equations}, Bulletin of the AMS,
  27(1) (1992), pp.~1--67.

\bibitem{DasDKS}
{\sc S.~R. Das, D.~Duffie, N.~Kapadia, and L.~Saita}, {\em Common failings: How
  corporate defaults are correlated}, The Journal of Finance, 62(1) (2007),
  pp.~93--117.

\bibitem{Finetti}
{\sc B.~De~Finetti}, {\em Su unimpostazione alternativa della teoria collettiva
  del rischio}, Transactions of the XVth International Congress of Actuaries, 2
  (1957), pp.~433--443.

\bibitem{Em75}
{\sc D.~C. Emanuel, J.~Michael~Harrison, and A.~J. Taylor}, {\em A diffusion
  approximation for the ruin function of a risk process with compounding
  assets}, Scandinavian Actuarial Journal, 4 (1975), pp.~240--247.

\bibitem{EvansPDE}
{\sc L.~Evans}, {\em Partial Differential Equations, 2nd Ed..}, American
  Mathematical Society, Providence, 2010.

\bibitem{Gerber}
{\sc H.~Gerber and E.~Shiu}, {\em Optimal dividends: analysis with brownian
  motion}, North American Actuarial Journal, 8(1) (2004), pp.~1--20.

\bibitem{Gerber72}
{\sc H.~U. Gerber}, {\em Games of economic survival with discrete and
  continuous income processes}, Operations Research, 20(1) (1972), pp.~37--45.

\bibitem{Grand}
{\sc J.~Grandell}, {\em A class of approximations of ruin probabilities},
  Scandinavian Actuarial Journal, sup1 (1977), pp.~37--52.

\bibitem{Grandits}
{\sc P.~Grandits}, {\em A two-dimensional dividend problem for collaborating
  companies and an optimal stopping problem}, Scandinavian Actuarial Journal,
  2019(1) (2019), pp.~80--96.

\bibitem{GuSZ}
{\sc J.~W. Gu, M.~Steffensen, and H.~Zheng}, {\em Optimal dividend strategies
  of two collaborating businesses in the diffusion approximation model},
  Mathematics of Operations Research, 43(2) (2017), pp.~377--398.

\bibitem{HojgaardT99}
{\sc B.~H. H{\o}jgaard and M.~Taksar}, {\em Controlling risk exposure and
  dividends payout schemes: insurance company example}, Mathematical Finance,
  9(2) (1999), pp.~153--182.

\bibitem{IbragimovJW10}
{\sc R.~Ibragimov, D.~Jaffee, and J.~Walden}, {\em Pricing and capital
  allocation for multiline insurance firms}, Journal of Risk and Insurance,
  77(3) (2010), pp.~551--578.

\bibitem{Ig}
{\sc D.~Iglehart}, {\em Diffusion approximations in collective risk theory},
  Journal of Applied Probability, 6 (1969), pp.~285--292.

\bibitem{Kaz3}
{\sc A.~E. Kyprianou, R.~Loeffen, and J.~L. P\'{e}rez}, {\em Optimal control
  with absolutely continuous strategies for spectrally negative l\'{e}vy
  processes}, Journal of Applied Probability, 49(1) (2012), pp.~150--166.

\bibitem{LoeffenR10}
{\sc R.~L. Loeffen and J.-F. Renaud}, {\em De {F}inetti's optimal dividends
  problem with an affine penalty function at ruin}, Insurance: Mathematics and
  Economics, 46(1) (2010), pp.~98--108.

\bibitem{MyersR01}
{\sc S.~C. Myers and J.~A. Read}, {\em Capital allocation for insurance
  companies}, Journal of Risk and Insurance, 68(4) (2001), pp.~545--580.

\bibitem{Perez}
{\sc K.~Noba, J.~L. P\'{e}rez, and X.~Yu}, {\em On the bail-out dividend
  problem for spectrally negative markov additive models}, SIAM Journal on
  Control and Optimization, 58(2) (2020), pp.~1049--1076.

\bibitem{Kaz2}
{\sc J.~L. P\'{e}rez and K.~Yamazaki}, {\em Refraction-reflection strategies in
  the dual model}, Astin Bulletin, 47(1) (2017), pp.~199--238.

\bibitem{Kaz}
{\sc J.~L. P\'{e}rez, K.~Yamazaki, and X.~Yu}, {\em On the bail-out optimal
  dividend problem}, Journal of Optimization Theory and Applications, 179(2)
  (2018), pp.~553--568.

\bibitem{PhillipsCA98}
{\sc R.~D. Phillips, J.~D. Cummins, and F.~Allen}, {\em Financial pricing of
  insurance in the multiple-line insurance company}, Journal of Risk and
  Insurance, 65 (1998), pp.~597--636.

\bibitem{Schmidli}
{\sc H.~Schmidli}, {\em Stochastic Control in Insurance}, Springer Verlag,
  2008.

\bibitem{TakadaS}
{\sc H.~Takada and U.~Sumita}, {\em Credit risk model with contagious default
  dependencies affected by macro-economic condition}, European Journal of
  Operational Research, 214(2) (2011), pp.~365--379.

\end{thebibliography}

\end{document}